\documentclass[11pt]{article}
\pdfoutput=1

\usepackage{geometry}
\usepackage{authblk}
\usepackage{bm}
\usepackage{amsmath}
\usepackage{amssymb}
\usepackage{amsthm}
\usepackage{hyperref}
\hypersetup{colorlinks,citecolor=blue,filecolor=blue,linkcolor=blue,urlcolor=blue}

\makeatletter
\usepackage{algorithm,algorithmic}
\usepackage{comment}
\usepackage[sort, numbers]{natbib}
\usepackage{microtype}
\usepackage{graphicx}
\usepackage{subcaption}
\usepackage{booktabs} 
\usepackage{float}

\let\hat\widehat
\newcommand{\goto}{\rightarrow}
\newcommand{\cov}{\boldsymbol{\Sigma}}
\newcommand{\bs}{\boldsymbol}
\newcommand{\usedim}{p}
\newcommand{\numobs}{n}
\newcommand{\defn}{:=}
\newcommand{\real}{\ensuremath{\mathbb{R}}}
\newcommand{\thetastar}{\ensuremath{\bm{\beta}}}
\newcommand{\inprod}[2]{\ensuremath{\langle #1 , \, #2 \rangle}}
\newcommand{\NORMAL}{\mathcal{N}}
\newcommand{\Ind}{\ensuremath{\textbf{I}}}
\newcommand{\mySigma}{\bm{\Sigma}}

\newtheorem{thm}{Theorem} 
\newtheorem{lemma}{Lemma}

\newtheorem{prop}{Proposition}
\newtheorem{example}{Example}
\newtheorem{definition}{Definition}
\newtheorem{remark}{Remark}

\title{Randomized tests for high-dimensional regression: A more efficient and powerful solution \footnotetext{This work has been accepted in part in Neural Information Processing Systems (NeurIPS) 2020.}}

\author[1]{Yue Li} 
\author[2]{Ilmun Kim}
\author[1]{Yuting Wei\thanks{Y.~Wei is supported in part by the NSF grant CCF-2007911 and DMS-2015447.}}
\affil[1]{Department of Statistics \& Data Science\\
	Carnegie Mellon University}
\affil[2]{Department of Pure Mathematics and Mathematical Statistics\\ 
	University of Cambridge}
\date{\today}

\begin{document}
\maketitle 

\begin{abstract}
	We investigate the problem of testing the global null in the high-dimensional regression models
	when the feature dimension $p$ grows proportionally to the number of observations $n$.
	Despite a number of prior work studying this problem, whether there exists a test that is model-agnostic, efficient to compute and enjoys high power, still remains unsettled.  
	In this paper, we answer this question in the affirmative by leveraging the random projection techniques, and propose a testing procedure that blends the classical $F$-test with a random projection step.
	When combined with a systematic choice of the projection dimension, the proposed procedure is proved to be minimax \emph{optimal} and, meanwhile, reduces the computation and data storage requirements. 
	We illustrate our results in various scenarios when the underlying feature matrix exhibits an intrinsic lower dimensional structure (such as approximate block structure or has exponential/polynomial eigen-decay), and it turns out that the proposed test achieves sharp adaptive rates.
	Our theoretical findings are further validated by comparisons to other state-of-the-art tests on the synthetic data.
\end{abstract}

\noindent \textbf{Keywords:} high-dimensional regression, random projections, $F$-test, 
minimax optimality, proportional regime

\setcounter{tocdepth}{2}
\tableofcontents

\section{Introduction}

Many applications in modern science and engineering operate in the regime where the number of parameters is comparable to the number of observations. For each dimension, on average, there are only a few samples that are available for statistical inference. This new aspect brings challenges to many traditional tools and methodologies which are often built upon the assumption that the number of observations per dimension goes to infinity. 

The current paper concentrates on the high-dimensional linear regression model which is the most commonly used statistical approach to model linear relationships between variables. 
One set of fundamental problems is to conduct hypothesis testing on the linear parameters, in particular, whether there are signals presented in the observations, and if there are, testing the statistical significance of each feature. 
Despite a considerable amount of prior work devoted to the study of this problem, whether an optimal testing procedure can be designed for an arbitrary feature matrix in the proportional regime ($n$ and $p$ grows proportionally) still remains unsettled. In particular, previous work often assumes that the feature space has an intrinsic low dimensional structure --- cases including assuming the parameter being sparse (e.g.~\cite{donoho2004higher,arias2011global,zhang2014confidence,van2014asymptotically,javanmard2014confidence}), 
lying in a $L_{p}$-ellipse or convex cone (e.g.~\cite{lehmann2006testing,ingster2012nonparametric,baraud2002non,wei2020local,wei2019geometry}) --- or assumes that the feature matrix is generated from a standard Gaussian design in the proportional regime (see, e.g.~\cite{sur2019likelihood,javanmard2014hypothesis,su2016slope,miolane2018distribution,celentano2020lasso}).

This paper aims to tackle this fundamental problem in the setting where \emph{(i)} no prior knowledge is assumed on the coefficient vector; \emph{(ii)} the feature dimension and the observation size grow proportionally with each other; \emph{(iii)} the fewest number of assumptions are imposed on the design matrix while being adaptive to the cases when the design matrix enjoys simpler structures. 
Formally, given $n$ i.i.d.~data pairs $(\bm{x}_1, y_1),(\bm{x}_2, y_2),\ldots,(\bm{x}_\numobs, y_\numobs)$, with $\bm{x}_i \in \real^{\usedim}$ and $y_i\in \real$, generated from a linear model, we have 
\begin{align*}
y_i = \inprod{\bm{x}_i}{\thetastar} + \sigma z_i,
\end{align*}
for some unknown vector $\thetastar \in \real^\usedim$, where $\sigma^2$ is the constant noise level, and $z_i$'s are independent from the covariates. Other assumptions shall be made clear in the next section.
Here $\inprod{\cdot}{\cdot}$ denotes the standard inner product.
The first step in diagnosis of linear models is to test the joint significance of all covariates, namely, test the hypothesis
\begin{equation}\label{eq:test}
H_0: \boldsymbol{\beta}=\boldsymbol{0} \quad \text{versus} \quad H_1: \boldsymbol{\beta}\neq \boldsymbol{0}.
\end{equation}

The classical and commonly used approach for the above testing problem is based on a global $F$-test \cite{lehmann2006testing}, which is known to be most powerful when the feature dimension $p$ is held fixed as the number of samples $n$ increases. 
However, it is also widely known that the $F$-test loses its power as the ratio $p/n$ increases when $p$ is comparable but smaller to $n$ (see e.g.~\cite{zhong2011tests}). 
Moreover, when the dimension exceeds the sample size, the $F$-test is not applicable anymore due to the singularity of the sample covariance matrix.
In these scenarios, \cite{zhong2011tests} proposed a test based on a $U$-statistic of fourth order and \cite{cui2018test} followed up with a modified test that improves power moderately. 
However, these proposed tests have low power and can be rather sub-optimal in many scenarios. 
We therefore ask the following question:
\begin{itemize}
	\item[]  \emph{Whether there exists an analogue of the classical $F$-test that is both efficient to compute and enjoys an optimal power in the regime where $p/n \in (0,\infty)$?}
\end{itemize}

In this paper, we answer this question in the affirmative using the random projection techniques.
Random projection, also known as ``sketching'', is based on the idea of 
storing only a (smartly designed) sketched version of the original data (often of lower dimensions), and performing learning algorithms on the sketched version.
It is now a standard technique for reducing data storage and computational costs, and promoting privacy. 
These advantages have motivated a lot of studies on designing efficient and differential-private estimation schemes (e.g. \cite{bingham2001random,liu2005random,sarlos2006improved,pilanci2017newton}), however, the statistical behaviors of test statistics based on random projections are not as well understood (e.g., 
\cite{cuesta2007random,jacob2010gains,clemenccon2009auc,NIPS2011_4260}). 
The closest to our work is by \cite{NIPS2011_4260} where a random-projection based Hotelling's $T$-test is proposed for testing the two-sample mean problem, given independent Gaussian samples.

Having set up the problem, we introduce a novel testing procedure, termed \emph{sketched $F$-test} for the testing problem~\eqref{eq:test}.
Intuitively, this test can be viewed as a two-step procedure: first each sample point $\bm{x}_i$ is projected to a $k$-dimensional random subspace for a preselected dimension $k$ and a random subspace $S_k$; and then a classical $F$-statistic is performed to check if the linear coefficients are zero between the projected data matrix and response vector.

\paragraph{Evaluations:} 
To evaluate the efficiency and asymptotic power of our proposed test, 
we compute its asymptotic power function and compare this with the state-of-the-art tests (\cite{zhong2011tests,cui2018test}) in terms of asymptotic relative efficiency (ARE). We demonstrate our sketched $F$-test is computationally efficient and has increased power in various scenarios. 
In the case when the design matrix has a lower intrinsic dimension, the case where the random projections are mostly powerful for, we show that it is sufficient to choose the sketched dimension proportional to the statistical dimension of the design matrix (defined in the sequel), without any loss of statistical accuracy. 
We also demonstrate the advantages of using our proposed test over its state-of-the-art competitors with higher asymptotic power and improved performance on synthetic data.

\subsection{Our contributions}
The main contributions of this paper are summarized below, all of which are built upon a careful analysis of a sketched version of the classical $F$-test.

\begin{itemize}
	\item In Section~\ref{sec:test}, we introduce a sketched $F$-test which does not restrain the size of $n,p$ as in prior literature.
	The explicit characterizations of its asymptotic power function are provided in Theorem~\ref{thm:power-function}. 
	Compared to prior work on this problem, our theoretical results are model-agnostic --- generalizing to scenarios that are other than Gaussian design and Gaussian noise.
	
	\item We provide a systematic way of selecting the projection dimension based on the 
	intrinsic dimension of the population design matrix $\mySigma$, which is current lacking in
	the state-of-art testing procedures concerning random projections. 
	When $\mySigma$ is indeed of lower dimensions, which is the case underlying most applications, our proposed test yields adaptive testing rates that are minimax optimal, and fully preserves the signal in the original model. These results are summarized in our Theorem~\ref{thm:optimality} and Theorem~\ref{thm:power-3}.
	
	\item We demonstrate the superiority of the proposed approach over existing state-of-the-art methods in terms of computational efficiency and asymptotic relative efficiency. We refer the readers to Section~\ref{subsec:are} for theoretical justifications and Section~\ref{Sec:sim} for experimental studies respectively.
	
\end{itemize}

\subsection{Other related work}

Related to the global testing problem, there has been an intensive line of work studying procedures that identify non-zero coefficients in high dimensional regression models.
To provide theoretical justifications, these procedures often pose more stringent assumptions on the model itself, such as sparsity (e.g.~\cite{zhang2014confidence,javanmard2014confidence,carpentier2018minimax}), independence or positive dependence between $p$-values (e.g.~\cite{benjamini2001control,gavrilov2009adaptive}).
In addition, the resulting theoretical guarantees mainly focus on type I error control, without a characterization of the statistical power (e.g.~\cite{barber2015controlling,candes2018panning,guo2014further}).
As a matter of fact, the global testing problem considered in this paper is often regarded as a first stage analysis in practice (see~\cite{wu2011rare, madsen2009groupwise}) and is intrinsically easier than testing individual coefficients. 
Therefore it allows us to derive a refined analysis of its power under much weaker assumptions. 


\paragraph{Notation.} We use $\overset{d}{=}$ for two random variables that have the same distribution. Let $\Phi(\cdot)$ denote the CDF of $\NORMAL(0,1)$, and $z_{\alpha}$ denote the upper $\alpha$ quantile of $\NORMAL(0, 1)$.  The upper $\alpha$ quantile of $F$-distribution with degrees of freedom $(p, n-p)$ is denoted by $q_{\alpha,p,n-p}$. 
Moreover, the norm $\|\cdot\|_2$ stands for Euclidean norm for a vector, and spectral norm for a matrix.
Matrix Frobenuis norm is denoted by $\|\cdot \|_F$. We call $a_n\asymp b_n$ if there is a  universal constant $c_0$ such that $\frac{1}{c_0} \leq \frac{a_n}{b_n} \leq c_0$ for sufficiently large $n$, and $a_n \lesssim b_n$ if $a_n \leq c_0 b_n$ for sufficiently large $n$.

\section{Sketched $F$-test}\label{sec:test}
In this section, we formally introduce the proposed sketched $F$-test and describe our main theoretical results along with some necessary background. We start our discussion by reviewing the classical $F$-test and describe why it fails in the high-dimensional setting.

\subsection{Classical $F$-test} \label{Section: Global F-test}

We find it useful to first formulate the observation model in the matrix form. 
Let $\bm{y} = (y_1,\ldots,y_\numobs)^\top$ and $\bm{X} \in \real^{\numobs \times \usedim}$ be the matrix with rows $\bm{x}_1^\top,\ldots, \bm{x}_\numobs^\top$, and we can write 
\begin{align}
\label{eq:model}
\bm{y} = \bm{X} \thetastar + \sigma \bm{z}.
\end{align}
Given i.i.d.~observations $\{\bm{x}_i, y_i\}_{i=1}^n$ from model \eqref{eq:model} with $n > p$, the classical $F$-test statistic is defined as
\[
F = \frac{ \hat{\boldsymbol{\beta}}^\top (\bm{X}^\top \bm{X})  \hat{\boldsymbol{\beta}}/p
}{\hat{\sigma}^2},
\]
where $\hat{\bs{\beta}} := (\bm{X}^\top\bm{X})^{-1}\bm{X}^\top \bm{y}$ is the least square estimator and $\hat{\sigma}^2 \defn \|\bm{y} - \bm{X}\hat{\boldsymbol{\beta}}\|_2^2/(n-p)$ is an unbiased estimator of $\sigma^2$. Under the null hypothesis $H_0$, it is well-known that the $F$-test statistic follows the $F$-distribution with $(p, n-p)$ degrees of freedom, whereas under the alternative $H_1$, it follows a noncentral $F$-distribution with $(p, n-p)$ degrees of freedom with the noncentrality parameter $\bs{\beta}^\top (\bm{X}^\top \bm{X}) \bs{\beta}/ 2\sigma^2$. In this setup, the $F$-test rejects the null hypothesis if $F \geq q_{\alpha, p, n-p}$ and its theoretical properties have been well-established in classical settings \cite[see,][]{rao2007linear}. 

As in \cite{zhong2011tests} (and also in \cite{baraud2002model,brown2002asymptotic,bayati2011lasso,donoho2016high}), we consider a tractable model where the design matrix $\bm{X}$ is randomly generated: each row of the design matrix $\bm{x}_i \in \real^{\usedim}$ is independently drawn from a multivariate distribution with covariance  $\mySigma$. 
We start by assuming the design matrix follows a multivariate Gaussian distribution with general covariance $\mySigma$ and then generalize some of our results to incorporate other random designs. Random designs enable us to carry out our analysis in a refined manner leveraging tools from random matrix theory and large deviation theory. 

Under this random design framework, it is easily seen that the power of the $F$-test is determined by the signal strength $\boldsymbol{\beta^\top \Sigma \beta}$, which is proportional to the expected value of the noncentrality parameter. When this signal strength is of constant order, the testing problem becomes trivial in a sense that the null and alternative hypotheses can be easily distinguished in the limit. This motivates us and others to consider the \emph{local alternative} in which $\boldsymbol{\beta^\top \Sigma \beta}$ diminishes as the sample size goes to infinity. Such framework is standard in asymptotic theory \cite[see, e.g.][]{van2000asymptotic} and has been considered by \cite{zhong2011tests,steinberger2016relative,cui2018test} among others. Under this local alternative, the following lemma studies the asymptotic power of the $F$-test in the high-dimensional regime where $p/n \rightarrow \delta \in (0,1)$. This result is the key ingredient to Theorem~\ref{thm:power-function} in which we study the asymptotic power of the sketched version of the $F$-test (see also \cite{steinberger2016relative}).

\begin{prop}\label{lem:f-small-p}
	Suppose the design matrix $\bs{X}$ is generated from a multivariate Gaussian distribution with covariance matrix $\cov$, and the noise vector $\bs{z}\sim \NORMAL(\bs{0}, \Ind_n)$ is independent of $\bs{X}$.
	Assume $\boldsymbol{\beta^\top \Sigma \beta}=o(1)$ and $\delta_n = p/n\goto \delta\in (0, 1)$ as $n \goto \infty$, then the power of the classical $F$-test which is defined as $\Psi_n^F \defn P(F \geq q_{\alpha,p,n-p})$, satisfies
	\[
	\Psi_n^F - \Phi\left(-z_{\alpha} + \sqrt{\frac{(1-\delta)n}{2\delta}} \frac{\boldsymbol{\beta}^\top \boldsymbol{\Sigma} \boldsymbol{\beta}}{\sigma^2}\right) \goto 0.
	\]
\end{prop}
In our Appendix~\ref{Section: the proof of lemma on small p}, we give a proof of this result which significantly simplifies the proof of Theorem 2.1 in \cite{steinberger2016relative} and serves as a building block for our other results. 

\subsection{The sketched $F$-test}
It is clear from Proposition~\ref{lem:f-small-p} that the classical $F$-test is not applicable when $p>n$ and performs poorly when $p/n$ is close to one. The main issue arises from the high variance of estimating $\mySigma^{-1}$ by inverting the sample covariance matrix. In this section, we tackle this problem by leveraging the random projections or sketching techniques. The usual purpose of sketching is to conduct dimension reduction while preserving the sample pairwise distances, however, this technique is employed here to reduce the variance in estimating the population covariance matrix. We witness a significant gain in testing power in various high-dimensional scenarios. 

Our testing procedure is summarized in the Algorithm~\ref{alg:example}. 
\begin{algorithm}
	\caption{~Sketched $F$-test}
	\label{alg:example}
	\begin{algorithmic}
		\STATE {\bfseries Input:} data matrix $\bm{X}\in\mathbb{R}^{n\times p}$, response vector $\bm{y}\in\mathbb{R}^n$, a sketching dimension $k < n$
		\STATE {\bfseries Output:} testing result for linear model~\eqref{eq:test}.
		\STATE {\bfseries Step 1:} generate a sketching matrix $S_k \in \mathbb{R}^{p\times k}$ with i.i.d. $\mathcal{N}(0, 1)$ entries;
		\STATE {\bfseries Step 2:} compute the least square regression estimate
		\begin{align*}
		\hat{\boldsymbol{\beta}}^S \defn (S_k^T \bm{X}^\top \bm{X} S_k)^{-1}S_k^\top \bm{X}^\top \bm{y}\in\mathbb{R}^k;
		\end{align*}
		\STATE {\bfseries Step 3:} calculate the sketched $F$-test statistic 
		\begin{equation}\label{eq:f-test}
		F(S_k) \defn \frac{\bm{y}^\top\bm{X}S_k \hat{\boldsymbol{\beta}}^S /k}{\|\bm{y}-\bm{X}S_k\hat{\boldsymbol{\beta}}^S\|_2^2/(n-k)};
		\end{equation}
		\STATE {\bfseries Step 4:} if $F(S_k) \geq q_{\alpha,k,n-k}$, reject $H_0$; otherwise accept $H_0$.
	\end{algorithmic}
\end{algorithm}

A few remarks are in order. First throughout the paper, we assume that the projection dimension $k < \min\{\text{rank}(\bm{X}), \text{rank}(\cov)\}$. With this choice, when $S_k$ has i.i.d.~$\mathcal{N}(0,1)$ entries, 
\begin{equation}\label{prop:almost-surely-invertible}
S_k^\top \bm{X}^\top \bm{X} S_k\text{\emph{ is invertible almost surely,}}
\end{equation}
which is formally proved in Appendix~\ref{Section: the proof of invertibility}.
This fact guarantees Algorithm~\ref{alg:example} is well-defined even when $p$ is much larger than $n$. Also note that under $H_0$, for any given $S_k$ and any realization $\bm{X}$, the sketched $F$-test statistic~\eqref{eq:f-test} follows the $F$-distribution with degrees of freedom $(k, n-k)$, and thus the proposed test is a valid level $\alpha$ test. 
We call $S_k$ from Algorithm~\ref{alg:example} a Gaussian sketching matrix with sketching dimension $k$.

In the following, we define a quantity that plays a key role in our further development, namely 
\begin{equation}\label{eq:delta}
\Delta_k^2 := \boldsymbol{\beta}^\top  \boldsymbol{\Sigma} S_k(S_k^\top \boldsymbol{\Sigma} S_k)^{-1} S_k^\top \boldsymbol{\Sigma}\boldsymbol{\beta}.
\end{equation}
It can be shown that $\Delta_k$ indeed determines the asymptotic power of the sketched $F$-test. 
As a consequence of Proposition~\ref{lem:f-small-p} applying to the sketched dataset, we establish the following guarantee on the asymptotic power of the sketched $F$-test.
\begin{thm}\label{thm:power-function}
	Suppose the design matrix $\bm{X}$ is generated from a multivariate Gaussian distribution with covariance $\mySigma$, and the noise vector $\bs{z}\sim \NORMAL(\bs{0}, \Ind_n)$ is independent of $\bs{X}$. Assume $\bs{\beta}^\top\bs{\Sigma}\bs{\beta}=o(1)$ and $\rho_n = k/n\goto\rho\in(0, 1)$ as $n\goto\infty$. Then, for almost all sequences of sketching matrix $S_k$, the power function of the sketched $F$-test, that is $\Psi^S_n(S_k) = P(F(S_k) > q_{\alpha,k,n-k})$, satisfies
	\begin{align} \label{eq:asymp}
	\Psi^S_n(S_k) - \Phi\left(-z_{\alpha} + \sqrt{\frac{(1-\rho)n}{2\rho}} \frac{ \Delta_k^2}{\sigma^2}\right) \goto 0.
	\end{align}
\end{thm}
We emphasize that the approximation of $\Psi^S_n(S_k)$ to the normal distribution function is precise in the limit including all constant factors. This asymptotic expression allows us to compare the proposed test with existing competitors in terms of the asymptotic relative efficiency in Section~\ref{subsec:are}. It is also helpful to notice that conditional on $S_k$, the sketched $F$-test can be simply viewed as the original $F$-test applied to the projected data set $\bm{X}S_k$. Now we are ready to provide the proof of Theorem~\ref{thm:power-function}.

\begin{proof}[Proof of Theorem~\ref{thm:power-function}.] 
	In order to complete the proof, we need to check the conditions in Proposition~\ref{lem:f-small-p} under the sketched regression setting. For a fixed $S_k$, by the property of a conditional Gaussian distribution, we have
	\[
	y_i | (\boldsymbol{X}_i'S_k) \sim N\left(
	\boldsymbol{\beta}^\top \boldsymbol{\Sigma} S_k (S_k^\top \boldsymbol{\Sigma} S_k)^{-1} S_k^\top \boldsymbol{X}_i, \nu^2\right),\]
	with $\nu^2 \defn \sigma^2 + \boldsymbol{\beta}^\top \boldsymbol{\Sigma}\boldsymbol{\beta}-\Delta_k^2$. Additionally let us write $\boldsymbol{\beta}^S \defn (S_k^\top \boldsymbol{\Sigma} S_k)^{-1} S_k^\top \boldsymbol{\Sigma} \boldsymbol{\beta}$. 
	Indeed Algorithm~\ref{alg:example} aims to test whether
	\begin{equation}\label{eq:hypothesis_1}
	H_0^S: \boldsymbol{\beta}^S = \boldsymbol{0} \quad \text{versus} \quad H_1^S:  \boldsymbol{\beta}^S \neq \boldsymbol{0},
	\end{equation}
	for the new regression model 
	\begin{equation}\label{eq:model_sketching}
	y_i = \boldsymbol{X}_i'S_k \boldsymbol{\beta}^S + z_i^S,
	\end{equation}
	where $z_1^S, \dots, z_n^S$ are independent random errors with $\text{Var}(z_i^S) =\nu^2$. Furthermore, when $S_k^\top \boldsymbol{\Sigma} S_k$ is invertible, the problem stated in \eqref{eq:hypothesis_1} becomes equivalent to testing whether
	$$
	H_0^S: S_k^\top \boldsymbol{\Sigma}\boldsymbol{\beta}=0 \quad \text{versus} \quad H_1^S: S_k^\top \boldsymbol{\Sigma}\boldsymbol{\beta}\neq 0.
	$$
	It is shown in \cite{NIPS2011_4260} that $\Delta_k^2 \leq \boldsymbol{\beta}^\top\boldsymbol{\Sigma} \boldsymbol{\beta}$. 
	Then we can show $\Delta_k^2=o(1)$ and $\nu^2 = \sigma^2 + o(1)$. Putting pieces together with Proposition~\ref{lem:f-small-p} completes the proof.
\end{proof}

\subsection{Testing power comparisons}
\label{subsec:are}

With the asymptotic power function characterized in expression \eqref{eq:asymp}, we compare our sketched $F$-test with other existing tests for the global null \eqref{eq:test} and highlight the advantages of our projection-based approach. 
First recall that \cite{zhong2011tests} introduced a test based on a fourth-order $U$-statistic and considered the local alternative $\boldsymbol{\beta}^\top \boldsymbol{\Sigma} \boldsymbol{\beta} = o(1)$ as well as other regularity conditions including 
\begin{equation}\label{eq:zc-condition}
\text{tr}(\bs{\Sigma}^4) = o(\text{tr}^2\{\bs{\Sigma}^2\}).
\end{equation}
We refer to their test as ZC test for simplicity. Under this asymptotic setting, the authors showed that the power of ZC test, denoted by $\Psi^{ZC}_n$ satisfies 
\begin{align}\label{eq:ZCpower}
\Psi^{ZC}_n := \Phi\left(-z_{\alpha} + \frac{n\|\boldsymbol{\Sigma}\boldsymbol{\beta}\|_2^2}{\sigma^2\sqrt{2\text{tr}(\boldsymbol{\Sigma}^2)}}\right).
\end{align}
As a follow-up to \cite{zhong2011tests}, \cite{cui2018test} proposed another $U$-statistic that improves the computational complexity of the statistic in \cite{zhong2011tests}, while achieving the same local asymptotic power~\eqref{eq:ZCpower}. 
Given these explicit power characterizations \eqref{eq:asymp}  and \eqref{eq:ZCpower}, the rest of this section is dedicated to comparing these methods in terms of a classical measure: the so called ``Asymptotic Relative Efficiency'', which is explained as follows.

\paragraph{Asymptotic Relative Efficiency (ARE):} In asymptotic statistics literature, a common way for comparing performances between different testing procedures is based on their 
asymptotic relative efficiency (ARE) \citep[see, e.g.][]{van2000asymptotic}
. 
Given two level $\alpha$ tests $\phi_1$ and $\phi_2$, the relative efficiency of $\phi_1$ to $\phi_2$ is defined as the ratio $n_2/n_1$ where $n_1$ and $n_2$ are the sample sizes required for $\phi_1$ and $\phi_2$ to achieve the same power against the same alternative. The ARE is then defined as the limiting value of this relative efficiency. Given this definition and building on the power expressions \eqref{eq:asymp} and \eqref{eq:ZCpower}, the ARE of ZC test to our sketched $F$-test is given by the limit of 
\begin{align}\label{eq:are}
\text{ARE}_n(\Psi^{ZC}_n;\Psi^{S}_n) :=
\frac{\sqrt{n}}{\sqrt{\text{tr}(\boldsymbol{\Sigma}^2)}} \bigg/ \sqrt{\frac{1-\rho}{\rho}}  \frac{\Delta_k^2}{\|\boldsymbol{\Sigma}\boldsymbol{\beta}\|_2^2}.
\end{align}
From the definition, it is clear that 
\begin{align*}
\text{ARE}_n(\Psi^{ZC}_n;\Psi^{S}_n) < 1 ~\Rightarrow~ \text{\bf sketched $F$-test is preferred}.
\end{align*}
The rest of this section aims to examine cases where $\text{ARE}_n(\Psi^{ZC}_n;\Psi^{S}_n) < 1$ with high probability. 
To facilitate our analysis, we first consider the case when the scaled $\bs{\beta}$ satisfies the assumption below. 
Remark that we adopt a frequentist approach throughout the paper, and this assumption is made only to illustrate the performance of the proposed testing procedure in an average sense.

\vskip .8em

\noindent \textbf{Assumption (A).} \emph{the normalized vector $\bs{\Sigma}^{1/2}\bs{\beta} / \|\bs{\Sigma}^{1/2}\bs{\beta}\|_2$ is uniformly distributed on the $p$-dimensional unit sphere, which is independent of $S_k$.}

\vskip .8em

Intuitively, Assumption~\textbf{(A)} holds when there is no preferred direction of the alternatives in which the scaled $\bs{\beta}$ differs from the zero vector. This assumption is standard when there is no prior information available for the alternatives. In particular, \cite{NIPS2011_4260} imposes a similar assumption in the context of two-sample mean testing.

Under Assumption~\textbf{(A)} and the regularity condition~\eqref{eq:zc-condition}, the next proposition provides an upper bound for $\text{ARE}_n(\Psi^{ZC}_n;\Psi^{S}_n)$ that holds with high probability. 
\begin{prop}\label{prop:are}
	Under the condition~\eqref{eq:zc-condition} and Assumption \textbf{(A)}, the following inequality holds with probability $1-o(1)$ as $n\goto \infty$:
	\begin{align} \label{eq: upper bound}
	\text{\emph{ARE}}_n(\Psi^{ZC}_n;\Psi^{S}_n)  \leq \frac{4}{\sqrt{\rho(1-\rho)}}  \frac{\text{\emph{tr}}(\boldsymbol{\Sigma})}{\sqrt{\text{\emph{tr}}(\boldsymbol{\Sigma^2})}}  \frac{1}{\sqrt{n}}.
	\end{align}
\end{prop}
We note that the condition~\eqref{eq:zc-condition} is required only for the ZC test, not for our sketched $F$-test. It essentially means that the eigenvalues of $\bs{\Sigma}$ should not decay too fast.  More importantly, if this eigenvalue condition is violated, then ZC test may not be a valid level $\alpha$ test even asymptotically. In such a case, it is not meaningful to compare the power of the given tests. In sharp contrast, the power expression \eqref{eq:asymp} for the sketched $F$-test holds regardless of the eigenvalue condition.

Building on Proposition~\ref{prop:are}, it is natural to ask the following questions. 

\paragraph{How do we select the dimension of $S_k$?} 
Notice that the upper bound of $\text{ARE}_n(\Psi^{ZC}_n;\Psi^{S}_n)$ is maximized when $\rho = 1/2$. In other words, when there is no extra information on the model, $k=\lfloor n/2\rfloor$ is a good choice of the sketching dimension to obtain higher asymptotic power. A similar choice was recommended by \cite{NIPS2011_4260} for the sketched version of Hotelling's $T^2$ test. In Section~\ref{sec:low-rank}, we explain how to leverage the information of $\bs{\Sigma}$ and further improve the power.

\paragraph{When~does~the~sketched~$F$-test~have~higher~power?} 
Notice that the upper bound in expression~\eqref{eq: upper bound} depends on $\text{tr}(\bs{\Sigma}) / \sqrt{\text{tr}(\bs{\Sigma^2})}$, which roughly measures the rate of decay of the eigenvalues. Recall that the sketched $F$-test is preferred over ZC test when $\text{ARE}_n(\Psi^{ZC}_n;\Psi^{S}_n)$ is less than one. That means, the sketched $F$-test becomes more favorable to ZC test in terms of power when $\text{tr}(\bs{\Sigma}) / \sqrt{\text{tr}(\bs{\Sigma^2})}$ grows slower than $1/\sqrt{n}$.


With the recommended choice of $k=\lfloor n/2\rfloor$, we make the upper bound~\eqref{eq: upper bound} be more concrete in the following two examples. Derivations of them can be found in Appendix~\ref{section:first two examples}.
\begin{example}\label{exp-1}
	For a positive sequence $t_n>0$, let $s:=s_n$ be another sequence that satisfies $\sqrt{s} \leq (1-\epsilon) t_n \sqrt{n} / 8$ for a small constant $\epsilon \in (0,1)$. Suppose that $\boldsymbol{\Sigma}$ can be well-approximated by a rank $s$ matrix in the sense that
	\begin{equation*}
	\lambda_1+ \dots + \lambda_s \geq (1-\epsilon) \cdot \text{\emph{tr}}(\boldsymbol{\Sigma}).
	\end{equation*}
	Then we have $\text{\emph{ARE}}_n(\Psi^{ZC}_n;\Psi^{S}_n)  \leq t_n$ with probability $1-o(1)$.
	In particular, if the covariance matrix is well-approximated by a $\sqrt{n}$ dimensional matrix, then one can take $t_n = 1/n^{1/4}$ and therefore 
	\begin{align*}
	\text{\emph{ARE}}_n(\Psi^{ZC}_n;\Psi^{S}_n)  \leq \frac{1}{n^{1/4}}.
	\end{align*}
\end{example}
\begin{example}\label{exp-2}
	When $\lambda_i(\bs{\Sigma}) \propto i^{-1/2}$ and $p/n \goto \rho$, we have that $\text{\emph{ARE}}_n(\Psi^{ZC}_n;\Psi^{S}_n) \leq 16\sqrt{\frac{\rho}{\log p}}$ with probability $1-o(1)$, which has a $0$ limit as $n,p\goto\infty$.
\end{example}

\section{Optimal guarantees for structured designs}
\label{sec:low-rank}

In this section, we discuss the case when $(\bs{\beta}, \cov)$ are not completely full dimensional, but instead have intrinsically lower dimensional structure, 
which is indeed the case underlying most applications (such as \cite{fienup1982phase,sundar2003skeleton,candes2010matrix,goldberg1992using,halko2011finding}). 
We first define our measure of intrinsic dimension $r$ in Section~\ref{sec:intrinsic-dim}. We then justify the optimality of choice $k=O(r)$ in two ways: in Section~\ref{sec:minimax}, we discuss the minimax optimality within the dim-$r$ class; in Section~\ref{sec:power}, we show that this choice fully preserves the signal strength and yields a non-random power expression of sketched $F$-test. In Section~\ref{sec:examples}, we demonstrate the consequences of the above results with several examples.

\subsection{Intrinsic dimensions}\label{sec:intrinsic-dim}
Let us start by introducing some notation to describe the spectral structure of $\bs{\Sigma}$. First denote the singular value decomposition of $\bs{\Sigma}$ as $\bs{\Sigma}=\bs{U}\bs{\Lambda}\bs{U}^\top$, where the diagonals of $\bs{\Lambda}$ are aligned in descending order. Write $\bs{\Lambda} = \text{diag}(\lambda_1, \dots, \lambda_p)$. Also define a rotated version of our signal as $\widetilde{\bs{\beta}} = \bs{U}^\top\bs{\beta}$. Note that $\bs{\Lambda}$ and $\widetilde{\bs{\beta}}$ describe variance and ``coefficient'' of the model in terms of orthogonalized feature dimensions, and are critical in determining the intrinsic structure of the model.

In the following, let us formally define the intrinsic dimension of the problem. 
\begin{definition} \label{def: intrinstic dim}
	We say model~\eqref{eq:model} has intrinsic dimension up to $r$, if we can find $\eta=o(1)$ and $r\leq p$, such that
	\begin{equation}
	\begin{aligned}\label{eq:intrinsic-dim}
	& \left(\frac{1}{r}\sum_{i=1}^r \widetilde{\beta}_i^2 \right) \cdot \left(\sum_{i=r+1}^p \lambda_i\right) + \sum_{i=r+1}^p \widetilde{\beta}_i^2 \lambda_i  \leq  \eta \bs{\beta}^\top\bs{\Sigma}\bs{\beta}; \\
	& \left( \frac{1}{r}\sum_{i=1}^r \widetilde{\beta}_i^2 + \frac{1}{p-r}\sum_{i=r+1}^p \widetilde{\beta}_i^2\right) \cdot r\lambda_{r+1} \leq \eta \bs{\beta}^\top\bs{\Sigma}\bs{\beta}.
	\end{aligned}
	\end{equation}
	Denote the collection of such $(\bs{\beta},\bs{\Sigma})$ as $\mathcal{D}(r)$.
\end{definition}
Here quantities $\eta$ and $r$ are both sequences of parameters indexed by $p$. 
To understand Definition~\ref{def: intrinstic dim}, suppose each $\widetilde{\beta}_i$ is of the same order, the above conditions boil down to
\begin{align*}
\sum_{i=r+1}^p \lambda_i \leq \eta \sum_{i=1}^p \lambda_i \quad \text{and} \quad 
r\lambda_{r+1} \leq \eta \sum_{i=1}^p \lambda_i.
\end{align*}
We refer readers to the examples in Section~\ref{sec:examples} for illustrations.

\subsection{Minimax optimality}\label{sec:minimax}
In the following, we describe our results under the classical minimax testing framework (\cite{ingster2012nonparametric}) and Gaussian design model~\eqref{eq:model}. Given observations $(\bm{X}, \bm{y})$, consider the problem of testing 
$$
H_0: \bs{\beta} = \bs{0} \quad \text{versus} \quad H_1: (\bs{\beta},\bs{\Sigma}) \in \Theta_r(\tau),
$$
in which the alternative space $\Theta_r(\tau)$ is specified by
$$\Theta_r(\tau) = \{ (\bs{\beta},\bs{\Sigma}) \in \mathcal{D}(r): \bs{\beta}^\top\bs{\Sigma}\bs{\beta} \geq \tau^2\},~~\text{for some}~\tau > 0.
$$
Note that the alternative is measured in terms of the Mahalanobis norm instead of the $\ell_2$ norm. We call $\psi$ a level-$\alpha$ test function, if $\psi$ is a measurable mapping from all possible values of $(\bm{X}, y)$ to $\left\lbrace 0, 1\right\rbrace$, and satisfy
$\mathbb{E}\left[\psi\vert H_0 \right] \leq \alpha$. The Type II error of $\psi$ and the minimax Type II error over $\Theta_r(\tau)$ are defined as
\[
r(\psi, \bs{\beta}, \bs{\Sigma}) :=  \mathbb{E}_{\bs{\beta},\bs{\Sigma}}\left[1 - \psi\right]; \quad \mathcal{R}_r(\tau) := \inf_{\psi} \sup_{(\bs{\beta},\bs{\Sigma})\in\Theta_r(\tau)} r(\psi, \bs{\beta},\bs{\Sigma}),
\]
in which the infimum is taken over all level-$\alpha$ test functions. We say level-$\alpha$ test $\psi^*$ is rate optimal with radius $\epsilon_n$, if for any $\gamma \in (\alpha, 1)$, there exists universal constant $c$, such that the following hold:
\begin{enumerate}
	\item[(i)](lower bound) when $\tau_n \leq c \epsilon_n$, $\mathcal{R}_r (\tau_n) \geq \gamma$,
	\item[(ii)](upper bound) 
	when $\tau_n/\epsilon_n \goto \infty$, for any $(\bs{\beta}_n,\bs{\Sigma}_n) \in \Theta_r(\tau_n)$, we have $r(\psi^*, \bs{\beta}_n,\bs{\Sigma}_n) = o_P(1)$.
\end{enumerate}

We are now able to state our optimality result in terms of the testing radius $\epsilon_n.$
\begin{thm}\label{thm:optimality}
	The sketched $F$-test is minimax rate optimal over $(\bs{\beta},\bs{\Sigma}) \in \mathcal{D}(r)$ with radius
	\begin{equation}
	\label{EqnRadius}
	\epsilon_n^2 = \frac{r^{1/2}}{n},
	\end{equation}
	and the upper bound is reached by choosing a sketching dimension $ar \leq k \leq br$ with any $1 < a \leq b$.
\end{thm}
\begin{remark} \normalfont
	The rate above is the same as the global testing rate of linear regression model with feature dimension $r$, sample size $n$, $\bs{\Sigma} = \ensuremath{\textbf{I}}_r$ and $r<n$ \citep[see ][]{carpentier2018minimax}. In this sense, the intrinsic dimension $r$ measures the minimal number of orthogonal dimensions needed to approximate the original model.
\end{remark}

\subsection{Refined power guarantees}\label{sec:power}
In the following, we provide a more precise characterization of the power function for our proposed test, with the sketching dimension $k$ chosen proportionally to the intrinsic dimension $r$.

\begin{thm}\label{thm:power-3}
	Suppose $(\bs{\beta},\bs{\Sigma})\in\mathcal{D}(r)$. Assume  $\bs{\beta}^\top\bs{\Sigma}\bs{\beta} = o(k/n)$. Then, for almost all sequences of sketching matrix $S_k$, the power function
	of the sketched $F$-test satisfies
	\[
	\Psi_n^S- \Phi\left(
	- z_{\alpha} + \frac{\sqrt{n}\bs{\beta}^\top\bs{\Sigma}\bs{\beta}}{\sigma^2} \sqrt{\frac{1-k/n}{2k/n}}
	\right) \overset{\text{p}}{\goto} 0.
	\]
\end{thm}
\noindent The proof is provided in the Appendix~\ref{section:theorem on power guarantee}.
Note that in Theorem~\ref{thm:power-3}, the quantity inside the second $\Phi$ function is completely deterministic. Now we make a few remarks with regard to the above result. 

\begin{remark} [\textbf{Choice~of~$k=O(r)$}]\normalfont
	Theorem~\ref{thm:power-3} shows that by choosing $k=O(r)$, we are able to obtain the same signal strength as the original model as $\Delta_k^2\leq\bs{\beta}^\top\bs{\Sigma}\bs{\beta}$; further increasing the sketching dimension does not increase the power. Together with Theorem~\ref{thm:optimality}, we justify the choice of $k=O(r)$ affirmatively.
\end{remark}
\begin{remark}[\textbf{Relaxation of Gaussian assumption}] \normalfont \label{remark: gaussian}
	So far we have assumed that the design matrix $\bs{X}$ and random errors $\bs{z}$ follow Gaussian distributions, mainly to simplify our presentation. Indeed, this Gaussian assumption is not necessary and can be replaced with more general conditions. To this end, we can leverage the recent results by \cite{steinberger2016relative} and prove that Theorem~\ref{thm:power-function} and Theorem~\ref{thm:power-3} hold under mild moment conditions on $\bs{X}$ and $\bs{z}$. Due to the space limit, we defer the technical details of this result to the Appendix~\ref{section:non-gaussian}.	
\end{remark}
\begin{remark}[\textbf{Computational Complexity}] \normalfont
	Note that the computational complexity of constructing $\bm{X}S_k$ is $O(nkp)$, which can be further reduced by fast Hadamard transform \citep[see, e.g.][]{yang2017randomized}. The complexity of the $F$-statistic is $O(k^2n)$ and thus the overall complexity of the sketched statistic becomes $O(npk)$. In contrast, the $U$-statistics in \cite{zhong2011tests} and \cite{cui2018test} have time complexity of $O(n^4+n^2p)$ and $O(n^2p)$, respectively. This illustrates a computational benefit of the proposed approach over the competitors especially when the design matrix has a low dimensional structure.
\end{remark}

\subsection{Examples}\label{sec:examples} 
We illustrate Theorem~\ref{thm:optimality} and Theorem~\ref{thm:power-3} in a range of different settings for parameters $(\bs{\widetilde{\beta}}, \cov)$. 
For simplicity, in the first three examples, we assume that each coordinate $\widetilde{\beta}_i$ is generated from the same distribution and study how intrinsic dimensions (as well as the testing rates) vary with different structures of $\mySigma.$
Throughout this section, $\eta$ is set to be $ 1/\log p$.
More details of the derivations can be found in the Appendix~\ref{section:details of second set of examples}.

\begin{example}[$\alpha$-polynomial decay]\label{example:polynomial}
	Suppose the eigenvalues decay as $\lambda_j \propto j^{-\alpha}$ with $\alpha >1$. Then we have
	\begin{align*}
	r \lesssim (\log p)^{\frac{1}{\alpha-1}}.
	\end{align*} 
\end{example}
\begin{example}[$\gamma$-exponential decay]\label{example:exponential}
	Suppose the eigenvalues decay as $\lambda_j \propto\exp(-j^{\gamma})$ with $\gamma > 0$. Then we have 
	\[ r \lesssim (\log \log p)^{\frac{1}{\gamma}}.
	\]
\end{example}
\begin{example}[block covariance]\label{example:block}
	Suppose the covariance matrix has $m$ identical blocks with $m\goto\infty$:
	\[
	\cov = \left(\begin{matrix}
	\bs{B} & & \\
	& \ddots & \\
	& & \bs{B}
	\end{matrix}\right), \qquad \text{where  }~ \bs{B} = (1-\rho)I_{d\times d} + \rho \bs{1}_d\bs{1}_d^\top, 
	\]
	with $p = md$. In this case, the spectrum of $\cov$ has $m$ copies of $(1-\rho)+\rho d$ and $p-m$ copies of $1-\rho$. Then as long as $\rho \geq 1 - \frac{c}{\log p}$, we have
	\[ r \lesssim m.
	\]
\end{example}
The last example shows how the structure of $\bs{\beta}$ changes the intrinsic dimension:
\begin{example}[structured coefficient]\label{example:beta-structure}
	Suppose the eigenvalues decay as $\lambda_j \propto j^{-1}$ and $0 < c_1 \leq \widetilde{\beta}_i\sqrt{i} \leq c_2$ for some fixed $c_1,c_2>0$. Then we have
	\[
	r \lesssim (\log p)^3.
	\] 
\end{example}
\vskip -0.1in
As can be seen from the above examples, our sketched procedure is adaptive to various structured designs. 
Note that the $\alpha$-polynomial and $\gamma$-exponential decays are classical setups for nonparametric estimation, and their effective dimensions under different problems, such as kernel regression estimation and optimization are well understood (e.g. \cite{yang2017randomized,wei2017early}). We would like to point out that the intrinsic dimensions we derived above are much smaller than that of the previous literature, echoing the common belief that testing is a much easier task than its corresponding estimation analogue.

\section{Simulation studies} 
\label{Sec:sim}

In this section, we present some empirical studies to validate our theoretical findings using synthetic data. 
We start by comparing the performance of our sketched $F$-test to the existing tests proposed by \cite{zhong2011tests} and \cite{cui2018test} under several scenarios. We refer to the latter two tests as ZC test and CGZ test, respectively. 
Our experiments are set in the same way as of \cite{zhong2011tests} and \cite{cui2018test}, to give a fair comparison to the other two methods. 
In the second part, we consider a sequence of testing problems indexed by $p$ with non-sparse signals. We demonstrate that the signal strength of the sketched coefficients $\Delta_k^2$ in \eqref{eq:delta} follows closely to that of the original coefficients $\bs{\beta}^\top \bs{\Sigma} \bs{\beta}$, leading to the power of the proposed test higher than the other competitors.

\paragraph{Power evaluation.}
We first demonstrate the empirical power performance of the tests by varying the decay rate of eigenvalues. For the slow-decay case, we consider $\bs{\Lambda}$ with $\lambda_i = \log^{-2}(i+1)$ and choose $k=\lfloor n/2\rfloor$. Whereas for the fast-decay case, we consider $\bs{\Lambda}$ with $\lambda_i = i^{-2/3}\log^{-1}(i+1)$ and choose $k=\lfloor 2 \log p\rfloor$. 

We sample each $\beta_i$ from $\text{Binomial}(3, 0.3) + 0.3 \mathcal{N}(0, 1)$ independently to generate heterogeneous and non-sparse signals; with each given $\bs{\Lambda}$, let $\bs{U}$ be column vectors of a QR decomposition of a $p\times p$ matrix with i.i.d. $\mathcal{N}(0,1)$ entries, and set $\cov = \bs{U}\bs{\Lambda}\bs{U}^\top$. We generate $\bs{Z}$ with i.i.d.~$\mathcal{N}(0,1)$ entries for the slow-decay case and i.i.d.~$t(2)$ entries for the fast-decay case, and calculate $\bs{X}=\cov^{1/2} \bs{Z}$. A noise vector $\bs{z}$ is similarly generated with i.i.d. $\mathcal{N}(0,1)$ entries. We scale $\bs{\beta}$ and $\cov$ to ensure that $\|\bs{\beta}\|_2 = c_1$ and $\|\cov\|_F = c_2$ for $c_1 \in \{0,1,5\}$ and $c_2 \in \{50,100,300\}$.

We present the Type I and Type II error rates in Table~\ref{table:combined} with $(n,p)=(50, 500)$. The simulations were repeated $500$ times to estimate the error rates. From the results, we first note that all of the tests control the Type I error rate at $\alpha=0.05$ reasonably well under the null. In terms of the Type II error, we see that the sketched $F$-test performs comparable to or better than CGZ and ZC tests in most of the considered scenarios. In particular, the sketched $F$-test has a clear advantage over the competitors for the fast-decay case, which coincides with our theory in Section~\ref{sec:low-rank}.

\begin{table}
	\caption{Comparison of Type I and Type II Error Rates.}
	\label{table:combined}
	\centering
	\scalebox{0.9}{
		\begin{tabular}{llllllll}
			\toprule
			&&\multicolumn{3}{c}{Slow-decay} & \multicolumn{3}{c}{Fast-decay}                  \\
			\cmidrule(r){3-5}\cmidrule(r){6-8}
			&&$\|\bs{\beta}\|_2=0(H_0)$ & $\|\bs{\beta}\|_2=1$ & $\|\bs{\beta}\|_2 =5$&$\|\bs{\beta}\|_2=0(H_0)$ & $\|\bs{\beta}\|_2=1$ & $\|\bs{\beta}\|_2 =5$ \\
			\cmidrule(r){1-2}\cmidrule(r){3-5}\cmidrule(r){6-8}
			&\textbf{Sketching}&$\bs{4.8\% }$& $14.4\%$ &$ \bs{0.2\%}$& $\bs{3.0\%} $& $\bs{2.6\%}$ &$ \bs{1.4\%}$     \\
			$\|\bs{\Sigma}\|_F =50$&CGZ&$5.2\%$ & $\bs{7.8\%}$ & $4.8\%$  & $5.6\%$ & $11.4\%$ & $10.8\%$    \\
			&ZC&$8.4\%$ & $18.9\%$ & $2.1\%$ & $6.3\%$ & $14.7\%$ & $18.9\%$ \\
			\cmidrule(r){1-2}\cmidrule(r){3-5}\cmidrule(r){6-8}
			&\textbf{Sketching}&$3.2\%$ & $\bs{1.4\%} $ &$\bs{0.0\%}$  & $\bs{4.0\%}$ & $\bs{1.4\%} $& $\bs{2.4\%}$  \\
			$\|\bs{\Sigma}\|_F =100$&CGZ&$6.0\%$ & $5.4\%$ & $4.6\%$  &$6.2\%$ & $10.4\%$ & $12.4\%$  \\
			&ZC&$\bs{2.1\%}$ & $16.8\%$ & $0.6\%$  & $4.2\%$ & $14.7\%$ & $6.3\%$  \\
			\cmidrule(r){1-2}\cmidrule(r){3-5}\cmidrule(r){6-8}
			&\textbf{Sketching} &$5.4\%$ & $\bs{0.0\%}$ & $\bs{0.0\%}$  & $4.6\%$ & $\bs{1.2\%}$ & $\bs{1.6\%}$  \\
			$\|\bs{\Sigma}\|_F =300$&CGZ&$5.4\%$ &$4.4\%$ & $3.0\%$ & $5.8\%$ &$12.4\%$ & $12.2\%$  \\
			&ZC&$\bs{4.2\%}$ & $6.3\%$ & $8.4\%$  & $\bs{4.2\%}$ & $8.4\%$ & $10.5\%$  \\
			\bottomrule
		\end{tabular}
	}
\end{table}

\paragraph{Asymptotic Behavior.}
In Section~\ref{sec:low-rank}, we discussed various cases for which the proposed sketched $F$-test is minimax rate optimal in power. This argument essentially relies on the observation that $\Delta_k^2 \overset{\text{p}}{\goto}\bs{\beta}^\top\cov\bs{\beta}$ within $\mathcal{D}(r)$ and $k=O(r)$. 
To illustrate the accuracy of this approximation, 
\begin{figure}[ht]
	\centering
	\begin{subfigure}[b]{0.5\textwidth}
		\centering
		\includegraphics[width=\columnwidth]{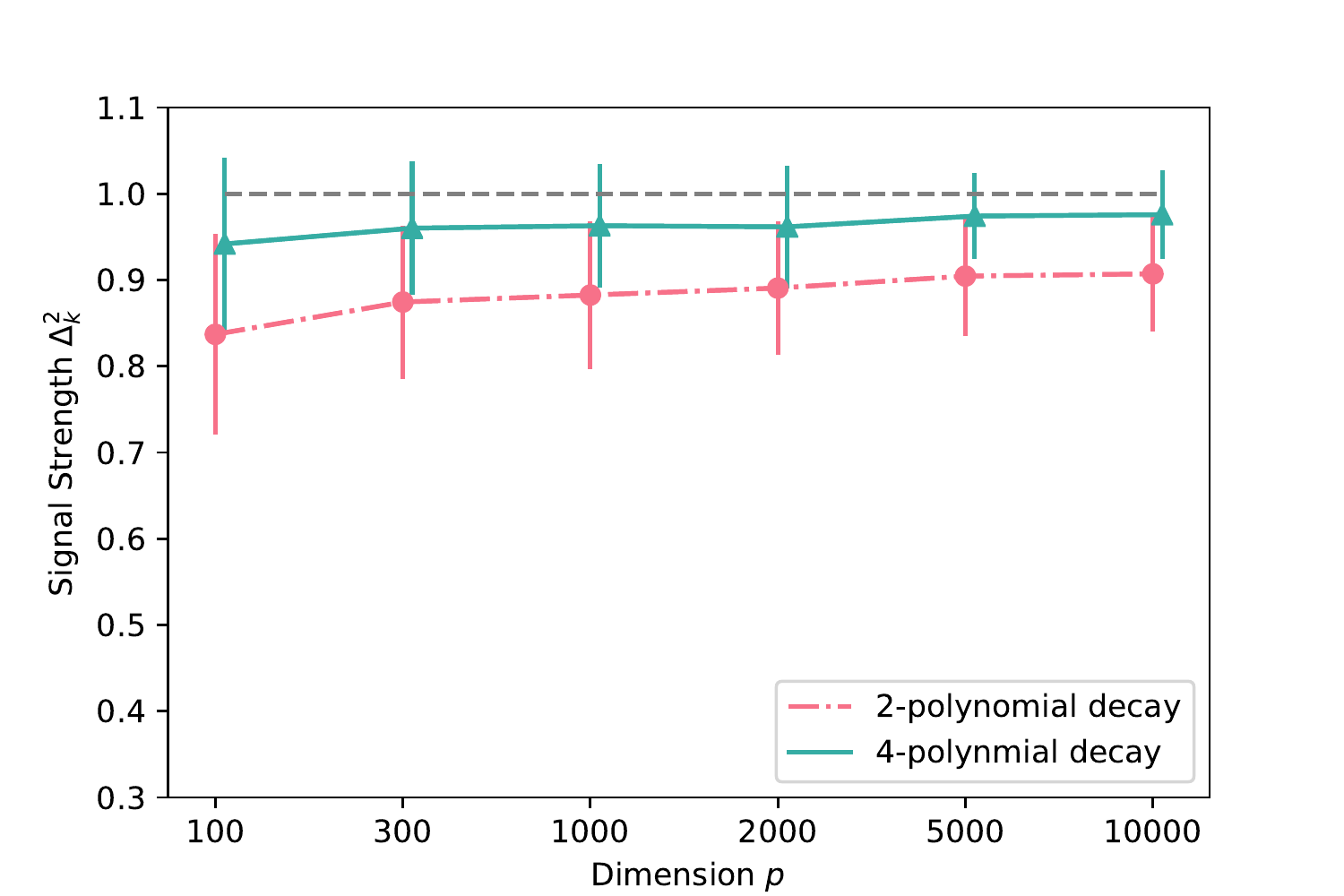}
		\caption{Stability of random signal.}
		\label{fig:signal}
	\end{subfigure}%
	\begin{subfigure}[b]{0.5\textwidth}
		\centering
		\includegraphics[width=\columnwidth]{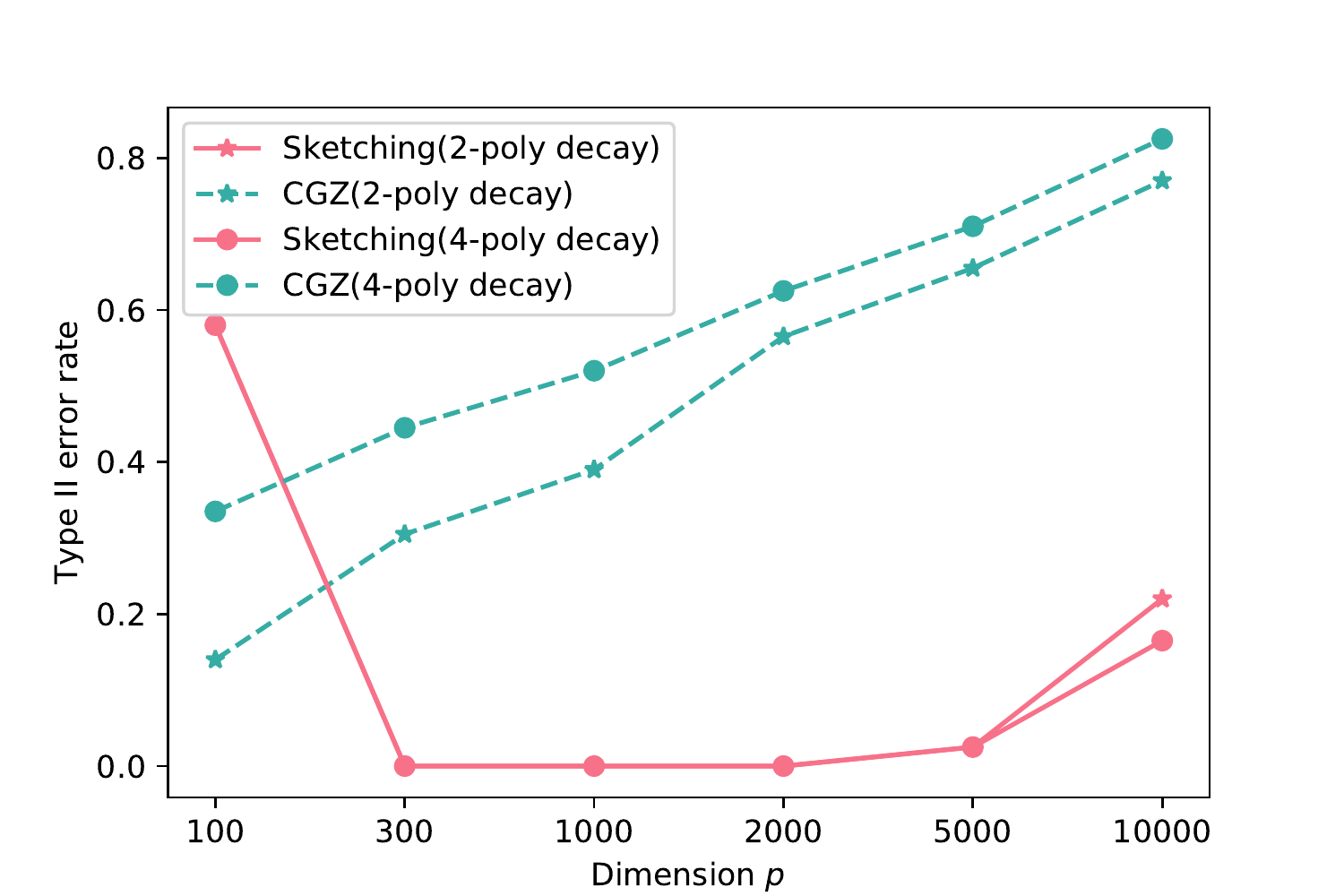}
		\caption{Error rate versus sample size.}
		\label{fig:error-plot}
	\end{subfigure}
	\caption{Simulation results with varying $p$.}
\end{figure}
we consider the polynomial-decay condition of eigenvalues in Example~\ref{example:polynomial} with the parameter $\alpha=2$ and $\alpha=4$. We then generate $\bs{\beta}$ and $\bs{U}$ in a similar manner to the previous simulation part, followed by a scaling step to ensure that $\bs{\beta}^\top\cov\bs{\beta}=1$. 

We first evaluate the random signal strength in the noiseless setting. Building on Example~\ref{example:polynomial}, we set the sketching dimension to be $k = \lfloor\min\{3 (\log p)^{\frac{1}{\alpha-1}}, n/2\}\rfloor$ and then calculate $\Delta_k^2$ for each $p\in \{100, 300, 1000, 2000, 5000, 10000\}$. Simulations were repeated 1000 times to obtain a confidence interval for $\Delta_k^2$ and the results are summarized in Figure~\ref{fig:signal}. From Figure~\ref{fig:signal}, we can see that the random signal $\Delta_k^2$ is fairly stable and follows closely to $\bs{\beta}^\top\cov\bs{\beta}$ over different sample sizes. This empirical result confirms that the random-projection approach maintains robust signal strength, which is a building block of our theoretical analysis in Section~\ref{sec:low-rank}.

To further demonstrate the performance of our approach, we generate $n=10\log^2 p$ samples for each scenario considered above with $\sigma^2=1$. When there is no prior knowledge about the population covariance structure, we take $k=\lfloor n/2\rfloor$ instead of the theoretical optimum. Under this setting, simulations were repeated $N=200$ times to approximate the type II error rates of our method and CGZ. The results are summarized in Figure~\ref{fig:error-plot}. From the results, we again observe a competitive performance of the proposed method over the competitor. We also found out that the ZC test behaves similarly to CGZ but computing their test statistic (fourth order U-statistic) is too expensive for large sample sizes. For this reason, we do not include the result of ZC test here.





\section{Key proof ingredients} 

The key of showing the upper bound part in Theorem~\ref{thm:optimality} and Theorem~\ref{thm:power-3} is a high-probability lower bound of the signal $\Delta_k^2$. Recall $\Delta_k^2 \leq \bs{\beta}^\top \cov \bs{\beta}$. Lemma~\ref{lem:minimax} below shows that, when $(\bs{\beta},\cov) \in \mathcal{D}(r)$ and sketching dimension is $O(r)$, the sketched model can capture most of signals in the original model. We state Lemma~\ref{lem:minimax} and its proof sketch here, while the details are deferred to Appendix~\ref{section:theorem 2} and \ref{section:theorem on power guarantee}.
\begin{lemma}\label{lem:minimax}
	When $(\bs{\beta},\cov) \in \mathcal{D}(r)$ and $ar \leq k \leq br$ with $1 < a \leq b$, 
	we have $\Delta_k^2 /\bs{\beta}^\top \cov \bs{\beta} \overset{p}{\goto} 1$.
\end{lemma}
\begin{proof}[Proof sketch of Lemma~\ref{lem:minimax}]
	First we introduce some additional notation. In the SVD decomposition $\bs{\Sigma}=\bs{U}\bs{\Lambda}\bs{U}^\top$, write $\bs{U}=\left[\bs{U}_r \ \bs{U}_{p-r}\right]$ and $\bs{\Lambda} = \left[\begin{matrix}\bs{\Lambda}_r & \\ & \bs{\Lambda}_{p-r} \end{matrix}\right]$, where $\bs{U}_r\in \mathbb{R}^{p\times r}$ and $\bs{\Lambda}_r \in \mathbb{R}^{r\times r}$. Then $\bs{\Sigma} = \bs{U}_r \bs{\Lambda}_r\bs{U}_r^\top + \bs{U}_{p-r} \bs{\Lambda}_{p-r} \bs{U}_{p-r}^\top:=\bs{\Sigma}_r + \bs{\Sigma}_{p-r}$. 
	
	The intuition comes from low-rank cases. If $\text{rank}(\bs{\Sigma}) = r$, using sketching dimension $k=r$ is enough. To see this, notice that when $\text{rank}(\bs{\Sigma}) = r$, we have $\text{rank}(\bs{\Sigma}S_k)=r$ almost surely, i.e., $\bs{\Sigma}S_k$ is of full-rank almost surely. Then $\exists \bs{\xi} \in\mathbb{R}^p$, such that 
	$
	\bs{\Sigma}\bs{\beta} = \bs{\Sigma}S_k\bs{\xi}.
	$
	It follows that $\Delta_k^2 = \boldsymbol{\beta}^\top  \boldsymbol{\Sigma} S_k(S_k^\top \boldsymbol{\Sigma} S_k)^{-1} S_k^\top \bs{\Sigma}S_k\bs{\xi} =\boldsymbol{\beta}^\top  \boldsymbol{\Sigma}\bs{\beta}$.
	
	In general case, we may not be able to find $\bs{\xi}$ satisfying $\bs{\Sigma}\bs{\beta} = \bs{\Sigma}S_k\bs{\xi}$, and we seek for some $\bs{\xi}$ to make the difference between $\bs{\Sigma}\bs{\beta}$ and $\bs{\Sigma}S_k\bs{\xi}$ small. Formally, as long as sketching dimension $k\geq r$, for any $\bs{\xi}$ that satisfies
	\begin{equation}\label{eq:restriction}
	\bs{U}_r^\top\bs{\beta} = \bs{U}_r^\top S_k\bs{\xi},
	\end{equation}
	define $\bs{\nu} = \cov_{p-r}(\bs{\beta}-S_k\bs{\xi})$. Then $\cov\bs{\beta}=\cov S_k \bs{\xi} + \bs{\nu}$, and 
	\begin{equation*}
	\begin{aligned}
	\Delta_k^2 = & \ (\bs{\xi}^\top S_k^\top \cov + \bs{\nu}^\top ) S_k(S_k^\top \cov S_k)^{-1} S_k^\top (\cov S_k\bs{\xi}  + \bs{\nu}) \\
	\geq & \ \bs{\beta}^\top \cov \bs{\beta} -(\bs{\beta}-S_k\bs{\xi})^\top \cov_{p-r} (\bs{\beta}-S_k\bs{\xi}),
	\end{aligned}
	\end{equation*}
	where the inequality follows by positive semi-definite property of $S_k(S_k^\top \boldsymbol{\Sigma} S_k)^{-1} S_k^\top$. 
	When $k\geq r$, we have $\text{rank}(\bs{U}_r^\top S_k)=r$ almost surely, so such $\bs{\xi}$ exists. To optimize the results, we seek for a solution of the problem 
	\[
	\mathrm{min}_{\bs{\xi}}\ (\bs{\beta}-S_k\bs{\xi})^\top \cov_{p-r}(\bs{\beta}-S_k\bs{\xi}) \quad s.t. \quad \bs{U}_r^\top (\bs{\beta}-S_k\bs{\xi})=0.
	\]
	The optimal $\bs{\xi}^*$ can be obtained by Lagrange multiplier. With Lagrange function
	$$
	\mathcal{L}(\bs{\xi}, \bs{\lambda}) = \frac{1}{2}(\bs{\beta}-S_k\bs{\xi})^\top \cov_{p-r}(\bs{\beta}-S_k\bs{\xi}) - \bs{\lambda}^\top \bs{U}_r^\top (\bs{\beta}-S_k\bs{\xi}),
	$$
	by solving the following two equations
	\begin{equation*}
	\frac{\partial \mathcal{L}(\bs{\xi}, \bs{\lambda})}{\partial \bs{\xi}} = 0  \quad \text{and} \quad
	\bs{U}_r^\top (\bs{\beta}-S_k\bs{\xi})=0,
	\end{equation*}
	we can solve for $\bs{\xi}^*$. The following lemma gives an upper bound on $(\bs{\beta}-S_k\bs{\xi}^*)^\top \cov_{p-r} (\bs{\beta}-S_k\bs{\xi}^*)$:
	\begin{lemma}\label{lem:delta-bound-1}
		Write $\widetilde{S}_1 = \bs{U}_r^\top S_k$, $\widetilde{S}_2 = \bs{U}_{p-r}^\top S_k$, $\widetilde{\bs{\beta}}_1 = \bs{U}_r^\top\bs{\beta}$ and $\widetilde{\bs{\beta}}_2 = \bs{U}_{p-r}^\top\bs{\beta}$. 
		With $\bs{\xi}^*$ defined above, we have
		$(\bs{\beta}-S_k\bs{\xi}^*)^\top \cov_{p-r} (\bs{\beta}-S_k\bs{\xi}^*) \leq 2L_1+2L_2$, where
		\begin{equation}\label{eq:delta-bound-1}
		\begin{aligned}
		L_1 &= \frac{\|\widetilde{S}_2^\top\bs{\Lambda}_{p-r} \widetilde{S}_2\|_2}{\lambda_{\min}(\widetilde{S}_1\widetilde{S}_1^\top)}\|\widetilde{\bs{\beta}}_1\|_2^2, \\
		L_2 &= \left( 1 + \kappa(\widetilde{S}_2^\top\bs{\Lambda}_{p-r} \widetilde{S}_2) \kappa(\widetilde{S}_1\widetilde{S}_1^\top) \right) \cdot \widetilde{\bs{\beta}}_2^\top \bs{\Lambda}_{p-r} \widetilde{\bs{\beta}}_2.
		\end{aligned}
		\end{equation}
		Here $\kappa(\cdot)$ represents the condition number of matrix, i.e., $\kappa(\bs{A}) = \lambda_{\max}(\bs{A})/\lambda_{\min}(\bs{A})$.
	\end{lemma}
	\vskip -0.1in
	To analyze the terms on the right hand side of \eqref{eq:delta-bound-1}, we need to characterize the behavior of the minimal and maximal eigenvalues of Gaussian random matrix. This can be done by covering argument; see, for example, \cite{vershynin2010introduction}. Details of the proof are deferred to Appendix~\ref{section:lemma on eigenvalues}, where we show that whenever $(\bs{\beta},\cov) \in \mathcal{D}(r)$ and $ar \leq k \leq br$, we have
	\[
	L_1 + L_2 \leq C\eta\bs{\beta}^\top \cov \bs{\beta}
	\]
	with probability at least $1-\exp(-c r)$, and thus the claim in Lemma~\ref{lem:minimax} follows.
\end{proof}

\section{Discussion}

In this work, we consider the problem of testing the overall significance for the regression coefficients in the high-dimensional settings. Building upon the random projection techniques, we introduce a sketched $F$-test for arbitrary dimension and sample size pair and develop theoretical gurantees for the proposed test including the asymptotic power and minimax optimality. We also demonstrate the advantages of the proposed test over the existing competitors in terms of the asymptotic relative efficiency and computational complexity. 
To our best knowledge, the proposed procedure is the first attempt to analyze in detail how sketching techniques work for testing regression coefficients. 

Our findings and analysis suggest a few directions for further investigations. 
For example, our procedure, as a general methodology, can be substantially extended to other testing problems. 
For instance, built upon an improved argument of the high-dimensional $F$-test (see \cite{steinberger2016relative}), our framework can be \emph{provably} adapted to testing whether $H_0: G \boldsymbol{\beta} = r_0$ or $H_1: G \boldsymbol{\beta} \neq r_0$ for matrix $G \in \mathbb{R}^{q \times p}$ and $r_0 \in \mathbb{R}^q$ with $q \leq p$.
In the case where, the joint significance of a group of coefficients are tested, it is sensible to combine a sketching step (over the complement set of features) with the classical $F$-test.
In addition, it would be interesting to see whether the sketching techniques can be adapted to other types of tests, apart from the $F$-test, as an effective approach for dimension reduction and statistical inference.


\vspace{1cm}
\bibliographystyle{alpha}
\newcommand{\etalchar}[1]{$^{#1}$}

\appendix

\section{Relaxation of the Gaussian assumptions}\label{section:non-gaussian}
In this part, we show that the proposed sketching test is still valid under more general conditions for both data matrix and noise distribution. To do this, we invoke a new set of assumptions on $\bs{X}_i$ and $z_i$ in model~\eqref{eq:model}, which hold beyond the Gaussian setting.

\paragraph{(B1)} The design vectors are generated as $\bs{x}_i = \bs{\Gamma}\bs{u}_i$, where $\bs{\Gamma}\in\mathbb{R}^{p\times m}$ satisfies $\bs{\Gamma}\bs{\Gamma}^\top = \bs{\Sigma}$ and $\bs{u}_1, \dots, \bs{u}_n$ are i.i.d. instances with $\mathbb{E}[\bs{u}_i] = \bs{0}$ and $\text{Var}[\bs{u}_i] = \Ind_m$ for some $m\leq k$. Additionally, we assume that $\bs{u}_i$ satisfies 

\quad \textbf{(a)} \emph{(polynomial tail)} There exists constant $c, C > 0$ such that for any $n\in\mathbb{N}$, orthogonal projection $P$ in $\mathbb{R}^m$ and $t > C \text{rank}(P)$, we have $\mathbb{P}(\|P \bs{u}_i\|^2 > t) \leq Ct^{-1-c}$;

\quad \textbf{(b)} \emph{(bounded moment)} We have $\sup_{\|v\|=1}(\mathbb{E}|v'\bs{u}_i|^8)^{1/8} = O(1)$ and for any symmetric matrix sequence $\bm{M}\in\mathbb{R}^{m\times m}$,
$$
\text{Var}[\bs{u}_i^\top \bm{M} \bs{u}_i] = O(\text{tr}(\bm{M}^2)) + o(\text{tr}^2(\bm{M})).
$$

\paragraph{(B2)} The noise vector $\bs{z}$ is independent of design matrix, with $\mathbb{E}[z_i^2] = 1$ and $\mathbb{E}[z_i^4] \leq c$ for $1\leq i\leq n$ and some universal constant $c>0$.

With this new set of assumptions, we are able to obtain similar results as in the Gaussian case. Theorem~\ref{thm:power-2} below, which builds on \cite{steinberger2016relative}, includes Theorem~\ref{thm:power-function} as a special case; we can also show Theorem~\ref{thm:power-3} holds if we replace the Gaussian assumptions of $\bs{X}$ and $\bs{z}$ with \textbf{(B1)} and \textbf{(B2)}.
\begin{thm}\label{thm:power-2}
	Besides \textbf{(B1)} and \textbf{(B2)}, assume $\lim\sup k/n <1$ and $\bs{\beta}^\top\bs{\Sigma}\bs{\beta} = o(k/n)$. Then, for almost all sequences of sketching matrix $S_k$, the power function $\Psi^S(S_k) = P\left\lbrace F(S_k) > q_{\alpha,k,n-k}\right\rbrace $ of test~\eqref{eq:f-test} satisfies
	\[
	\Psi_n^F - \Phi\left(
	- z_{\alpha} + \frac{\sqrt{n}\Delta_k^2}{\sigma^2} \sqrt{\frac{1-k/n}{2k/n}}
	\right) \goto 0.
	\]
\end{thm}
The proof of the result shares the same spirit as the proof of Theorem~\ref{thm:power-function}; one major difference is that, when the design matrix is not Gaussian, sketched noise $z_i^S$ is not independent of sketched data $S_k\bs{X}_i$ anymore, requiring extra efforts to characterize the behavior of $F(S_k)$. We list some technical details in Section~\ref{section:additional-proof}.

\paragraph{Remark:} We note that the assumptions \textbf{(B1)} and \textbf{(B2)} are mild. The moment and tail conditions hold for a wide range of random instances beyond Gaussian, including heavy-tailed ones such as log-normal distribution. Also note that we do not require entries of $\bs{u}_i$ to be independent with each other.

\section{Proofs of main results and other details}\label{section:main-proof}

\subsection{Proof of Proposition~\ref{lem:f-small-p}} 
\label{Section: the proof of lemma on small p}
First, let us write the second term inside $\Phi(\cdot)$ as 
\begin{equation}\label{eq:eta}
\eta = \sqrt{\frac{(1-\delta)n}{2\delta}} \frac{\boldsymbol{\beta}^\top \boldsymbol{\Sigma} \boldsymbol{\beta}}{\sigma^2}.
\end{equation}
We also define
\[
\hat{\sigma}^2 = \frac{\bs{y}^\top (\bs{\Ind}_p - \bs{X}(\bs{X}^\top \bs{X})^{-1}\bs{X}^\top)\bs{y}}{n-p} \quad \text{and} \quad T = \frac{\hat{\sigma}^2}{\sigma^2}\sqrt{\frac{n\delta(1-\delta)}{2}}(F-1).
\]
The proof builds on the following two claims, which are proved at the end of this section.
\begin{align} \label{eq:power-2}
&\sqrt{n}\left(\frac{\hat{\sigma}^2}{\sigma^2}-1\right) =O_P(1) \quad \text{and} \\[.5em]
&T - \eta \xrightarrow{d} \mathcal{N}(0, 1). \label{eq:power-1}
\end{align}
We now continue the main line of the proof assuming the claims in~\eqref{eq:power-2} and \eqref{eq:power-1} hold.
By the claim~\eqref{eq:power-2} we know $\hat{\sigma}^2 /\sigma^2 \xrightarrow{p} 1$. 
Note that $\eta = o(\sqrt{n})$ under local alternative assumption. By Slutsky's theorem, 
\begin{equation}\label{eq:power-0}
G\defn \sqrt{\frac{n\delta(1-\delta)}{2}}(F -1) - \eta = \frac{\sigma^2}{\hat{\sigma}^2} \left(T - \eta\right) + \left( \frac{\sigma^2}{\hat{\sigma}^2}-1\right) \eta \xrightarrow{d} \mathcal{N}(0, 1).
\end{equation}
We can use the convergence result~\eqref{eq:power-0} to show the claim in Lemma~\ref{lem:f-small-p}. Additionally write 
\begin{align} \label{eq: definition of s}
s \defn \sqrt{\frac{n\delta(1-\delta)}{2}}(q_{\alpha,p,n-p}-1).
\end{align}
Notice that $\Phi(\cdot)$ is Lipschitz-$1$ and thus we have
\begin{align*}
\left| \Psi_n^F - \Phi(-z_{\alpha} + \eta ) \right|  = &  \ \left| \mathbb{P}\left( G \geq  s-\eta \right) - \Phi(-z_{\alpha} + \eta)\right| \\[.5em]
\overset{\text{(i)}}{\leq} & \ \left| \mathbb{P}\left( G  \leq s-\eta \right) - \Phi\left(s-\eta\right)\right|  + \left| \Phi\left(s-\eta\right) -\Phi(z_{\alpha} - \eta) \right| \\[.5em]
\overset{\text{(ii)}}{\leq} &\  \sup_{x\in\mathbb{R}} \left| \mathbb{P}\left( G \leq x \right) - \Phi\left(x\right)\right|  + \left| s-z_{\alpha}\right|,
\end{align*}
where step~(i) uses the fact $\Phi(x)=1-\Phi(-x))$ and step~(ii) uses Lipschitz property of $\Phi$. To analyze the second term, we need Lemma 2.1 of \cite{bai1996effect} which provides an approximation of $q_{\alpha, p, n-p}$ when $p = \delta n$ for $\delta\in(0, 1)$.
\begin{lemma}[Lemma 2.1 of \cite{bai1996effect}]
	\label{lem:bai}
	When $p = \delta n$ with $\delta\in(0, 1)$, we have
	\begin{align*}
	q_{\alpha, p, n-p} = 1 + \sqrt{\frac{2}{n\delta(1-\delta)}} z_{\alpha} + o(n^{-1/2}).
	\end{align*}
\end{lemma}
Rearranging the statement of Lemma~\ref{lem:bai} yields $s = z_{\alpha} + o(1)$ where $s$ is defined in \eqref{eq: definition of s}. We also know $\sup_{x\in\mathbb{R}} \left| \mathbb{P}\left( G \leq x \right) - \Phi\left(x\right)\right| \goto 0$ by the approximation~\eqref{eq:power-0}. Combining these pieces yields $\left| \Psi_n^F - \Phi(-z_{\alpha} + \eta ) \right|=o(1)$ and thus Proposition~\ref{lem:f-small-p} follows.

\subsubsection*{Proof of Claim~\eqref{eq:power-2}}
Write $\bs{H}=\bs{X}(\bs{X}^\top \bs{X})^{-1}\bs{X}^\top$.
Notice that $\bm{HX} = \bm{X}$ and then $(\Ind_p - \bs{H})\bs{X}\bs{\beta}=\bs{0}$. By the linearity assumption $\bs{y} = \bs{X}\bs{\beta}+\sigma \bs{z}$, we can write 
\begin{equation}\label{eq:sigma-hat-2}
\frac{\hat{\sigma}^2}{\sigma^2} 
= \frac{(\bs{X}\bs{\beta}+\sigma \bs{z})^\top (\Ind_p - \bs{H})(\bs{X}\bs{\beta}+\sigma \bs{z})}{(n-p)\sigma^2} 
= \frac{1}{n-p} \bs{z}^\top (\Ind_p - \bs{H}) \bs{z}.
\end{equation}
Additionally, by our model assumption, the noise vector $\bs{z}\sim \mathcal{N}(\bs{0}, \Ind_p)$ is independent of $\bs{X}$. For any given $\bs{X}$ with rank $p$, $\Ind_p - \bs{H}$ is a projection matrix with rank $(n-p)$, and in this case $\bs{z}^\top (\Ind_p - \bs{H}) \bs{z}\vert \bm{H} \sim \chi^2_{n-p}$.
Under the Gaussian setting, we know $\text{rank}(\bs{X})=p$ almost surely, so $\hat{\sigma}^2/\sigma^2 \overset{d}{=} \chi^2_{n-p}/(n-p)$. Recall that $p=\delta n$, and thus $\sqrt{n}\left(\chi^2_{n-p}/(n-p) - 1\right) = O_P(1)$, which in turn leads to $\sqrt{n}\left( \hat{\sigma}^2/ \sigma^2-1\right) =O_P(1)$. This completes the proof of claim~\eqref{eq:power-2}.


\subsubsection*{Proof of Claim~\eqref{eq:power-1}}
We first rearrange the expression of $T$ in \eqref{eq:power-1}. By definition of $T$ in \eqref{eq:power-1}, we have
$$
T = \frac{\hat{\sigma}^2}{\sigma^2}\sqrt{\frac{n\delta(1-\delta)}{2}}(F-1) = \frac{\hat{\sigma}^2}{\sigma^2}\sqrt{\frac{n\delta(1-\delta)}{2}} \left(\frac{\bm{y}^\top\bm{H}\bm{y}/p}{\hat{\sigma}^2}-1\right) = \sqrt{\frac{n\delta(1-\delta)}{2}} \left(\frac{\bm{y}^\top\bm{H}\bm{y}/p}{\sigma^2}-\frac{\hat{\sigma}^2}{\sigma^2}\right).
$$
Using the fact that $\bm{HX} = \bm{X}$, we have
\[
\bm{y}^\top\bm{H}\bm{y} = (\bs{X}\bs{\beta}+\sigma \bs{z})^\top \bm{H}(\bs{X}\bs{\beta}+\sigma \bs{z}) = \sigma^2 \bm{z}^\top\bm{H}\bm{z} + 2\sigma \bm{\beta}^\top\bm{X}^\top\bm{z} + \bs{\beta}^\top \bs{X}^\top \bs{X} \bs{\beta}.
\]
Combining the above with another expression of $\widehat{\sigma}^2/\sigma^2$ in \eqref{eq:sigma-hat-2}, we can write $T$ as 
\[
T = \sqrt{\frac{n\delta(1-\delta)}{2}} 
\left(\frac{\bs{z}^\top\bs{H}\bs{z}}{p} - \frac{\bs{z}^\top(\Ind_p-\bs{H})\bs{z}}{n-p}  + \frac{\bs{\beta}^\top \bs{X}^\top \bs{X} \bs{\beta}}{p\sigma^2} + \frac{2}{\sigma} \frac{\bm{\beta}^\top\bm{X}^\top\bm{z}}{p}\right).
\]
By recalling $\eta$ defined in \eqref{eq:eta}, we can decompose $T-\eta$ as $T-\eta = T_1 + (T_2 - \eta) + T_3$, where
\begin{align*}
& T_1 = \sqrt{\frac{n\delta(1-\delta)}{2}}\left(
\frac{\bs{z}^\top\bs{H}\bs{z}}{p} - \frac{\bs{z}^\top(\Ind_p-\bs{H})\bs{z}}{n-p}
\right), \\
& T_2-\eta = \eta \left(\frac{\bs{\beta}^\top \bs{X}^\top \bs{X} \bs{\beta}}{n\bs{\beta}^\top \cov \bs{\beta}} - 1\right) \quad \text{and} \\ 
& T_3 =  \frac{1}{\sigma} \sqrt{\frac{2(1-\delta)}{n\delta}}\bs{\beta}^\top\bs{X}^\top\bs{z}.
\end{align*}
In what follows, we prove $T_1 \xrightarrow{d} \mathcal{N}(0,1)$, $T_2-\eta \xrightarrow{d} 0$ and $T_3 \xrightarrow{d} 0$ and thus $T - \eta \xrightarrow{d} \mathcal{N}(0,1)$ as desired.

\paragraph*{Analyzing $T_1$:} Note that $\bs{H}=\bs{X}(\bs{X}^\top \bs{X})^{-1}\bs{X}^\top$ is a projection matrix with rank $p$ almost surely. Therefore, conditional on $\bs{H}$, we have $\bs{z}^\top\bs{H}\bs{z}\vert\bs{H} \sim\chi^2_p$ and $\bs{z}^\top(\bs{I}-\bs{H})\bs{z}\vert\bs{H} \sim\chi^2_{n-p}$ and these are independent to each other. By letting $\omega_1, \omega_2\overset{iid}{\sim}\mathcal{N}(0,1)$, we may apply the central limit theorem and see that
\begin{align*}
& \bs{z}^\top\bs{H}\bs{z}/p\vert\bs{H} = 1 + \omega_1/\sqrt{p} + o_P(n^{-1/2}), \\[.5em]
& \bs{z}^\top(\bs{I}-\bs{H})\bs{z}/(n-p)\vert\bs{H} = 1 + \omega_2/\sqrt{n-p} + o_P(n^{-1/2}).
\end{align*}
Then we conclude that $T_1\vert \bs{H} \xrightarrow{d} \mathcal{N}(0,1)$ and thus $T_1 \xrightarrow{d} \mathcal{N}(0,1)$ as well by dominated convergence theorem. 

\paragraph*{Analyzing $T_2$:} Since $\bs{X}\bs{\beta} \sim \mathcal{N}(\bs{0}, (\bs{\beta}^\top\cov\bs{\beta}) \Ind_p)$ under the Gaussian setting, it follows that
\begin{align*}
\eta \left(\frac{\bs{\beta}^\top \bs{X}^\top \bs{X} \bs{\beta}}{n\bs{\beta}^\top \cov \bs{\beta}} - 1\right) \overset{d}{=} \eta \left(\frac{\chi^2_n}{n}-1 \right).
\end{align*}
Together with observations (i)~$\eta =o(\sqrt{n})$ and (ii)~$\sqrt{n}(\chi^2_n/n-1) = O_P(1)$, we conclude $T_2-\eta \xrightarrow{d} 0$.

\paragraph*{Analyzing $T_3$:} To show $T_3 \xrightarrow{d} 0$, it suffices to prove $\bs{\beta}^\top\bs{X}^\top\bs{z} = o_P(\sqrt{n})$. By the independence between $\bs{X}$ and $\bs{z}$, we have $\mathbb{E}\left[\bs{\beta}^\top\bs{X}^\top\bs{z}\right] = 0$ and $\text{Var}(\bs{\beta}^\top\bs{X}^\top\bs{z}) = \mathbb{E}\left[ \bs{\beta}^\top\bs{X}^\top\bs{z} \bs{z}^\top \bs{X}\bs{\beta}\right] = \mathbb{E}\left[ \bs{\beta}^\top\bs{X}^\top\bs{X}\bs{\beta}\right] = n \bs{\beta}^\top \cov \bs{\beta} = o(n)$. Therefore $\bs{\beta}^\top\bs{X}^\top\bs{z} = o_P(\sqrt{n})$ holds.

Combining the results, we complete the proof of claim~\eqref{eq:power-1}.


\subsection{Proof of claim~\eqref{prop:almost-surely-invertible}} \label{Section: the proof of invertibility}
Since $\text{rank}(\bs{A}^\top \bs{A}) = \text{rank}(\bs{A})$ for any matrix $\bs{A}$, we observe $\text{rank}(S_k^\top\bs{X}^\top\bs{X}S_k) = \text{rank}(\bs{X}S_k)$. For any realization of $\bs{X}$ with no all-zero rows, the entries of $\bs{X}S_k$ are independent Gaussian random variables and thus $\bs{X}S_k$ has full-rank $k$. By construction, $\bs{X}$ does not have all-zero rows almost surely, and thus $\text{rank}(S_k^\top\bs{X}^\top\bs{X}S_k) = k$ almost surely.


\subsection{Proof of Proposition~\ref{prop:are}} 
Rearranging expression~\eqref{eq:are} in the main text, we have
\[
\text{ARE}_n(\Psi^{ZC}_n;\Psi^S_n)= \left(\frac{4}{\sqrt{\rho(1-\rho)}}  \frac{\text{tr}(\boldsymbol{\Sigma})}{\sqrt{\text{tr}(\boldsymbol{\Sigma^2})}}  \frac{1}{\sqrt{n}}\right) \cdot \left(
\frac{\bs{\beta}^\top\bs{\Sigma}\bs{\beta}}{\Delta_k^2}  \frac{k}{2p}
\right)\cdot \left(
\frac{\|\bs{\Sigma\beta}\|^2}{\bs{\beta}^\top\bs{\Sigma}\bs{\beta}}  \frac{p}{2\text{tr}(\bs{\Sigma})}
\right),
\]
where we recall that 
\begin{align*}
\Delta_k^2 := \boldsymbol{\beta}^\top  \boldsymbol{\Sigma} S_k(S_k^\top \boldsymbol{\Sigma} S_k)^{-1} S_k^\top \boldsymbol{\Sigma}\boldsymbol{\beta}.
\end{align*}

The first term is exactly what we want; it remains to derive high-probability bounds for the second and third terms. Define
\[
\mathcal{E}_1 = \left\lbrace \frac{\Delta_k^2}{\bs{\beta}^\top\bs{\Sigma}\bs{\beta}} \geq \frac{k}{2p} \right\rbrace \quad \text{and} \quad \mathcal{E}_2 = \left\lbrace \frac{\|\bs{\Sigma\beta}\|^2}{\bs{\beta}^\top\bs{\Sigma}\bs{\beta}} \leq\frac{2\text{tr}(\bs{\Sigma})}{p} \right\rbrace .
\]
If we can show $\mathbb{P}(\mathcal{E}_1)\goto 1$ and $\mathbb{P}(\mathcal{E}_2)\goto 1$ as $n\goto \infty$, the claim of Proposition~\ref{prop:are} follows.

The remaining parts of the proof rely on concentration bounds of Gaussian quadratic forms. See Lemma 0.2. in \cite{bechar2009bernstein} for the proof of the following lemma:
\begin{lemma}[\cite{bechar2009bernstein}]\label{lem:gaussian-concentration}
	For any symmetric matrix $\bs{A}\in\mathbb{R}^{p\times p}$ with $\bs{A}\succeq 0$, $\bs{Z}\sim \mathcal{N}(0, I_{p\times p})$ and any $t>0$, we have
	\begin{align*}
	& \mathbb{P}\left(\bs{Z}^\top \bs{A} \bs{Z} \geq \emph{tr}(\bs{A}) + 2 \|\bs{A}\|_F \sqrt{t} + 2 \|\bs{A}\|t \right) \leq \exp(-t) \quad \text{and} \\[.5em]
	& \mathbb{P}\left(\bs{Z}^\top \bs{A} \bs{Z} \leq \emph{tr}(\bs{A}) - 2 \|\bs{A}\|_F \sqrt{t} \right) \leq \exp(-t).
	\end{align*}
\end{lemma}
We also state the useful matrix inequality used in the proof:
\begin{lemma}\label{lem:matrix-norm}
	For a symmetric matrix $\cov \in\mathbb{R}^{p\times p}$ and $\cov\neq \bs{0}$, we have $$\frac{\emph{tr}(\cov)}{\|\cov\|_F} \geq \left(\frac{\emph{tr}^2(\cov^2)}{\emph{tr}(\cov^4)}\right)^{1/8}. $$
\end{lemma}
The proof of Lemma~\ref{lem:matrix-norm} can be found in Section~\ref{SecPfMatrix-norm}. Using Lemma~\ref{lem:gaussian-concentration}, we first show $\mathbb{P}(\mathcal{E}_1)\goto 1$. By assumption \textbf{(A)}, we can write $\bs{\Sigma}^{1/2}\bs{\beta} / \|\bs{\Sigma}^{1/2}\bs{\beta}\|_2$ as $\bs{Z}/\|\bs{Z}\|_2$, where $\bs{Z}\sim \mathcal{N}(\bs{0}, \bs{I}_p)$. Then
\[
\frac{\Delta_k^2}{\bs{\beta}^\top\bs{\Sigma}\bs{\beta}} = \frac{1}{\|\bs{Z}\|_2^2} \bs{Z}^\top \bs{\Sigma}^{1/2}S_k(S_k^\top \boldsymbol{\Sigma} S_k)^{-1} S_k^\top \boldsymbol{\Sigma}^{1/2}\bs{Z} \defn \frac{1}{\|\bs{Z}\|_2^2} \bs{Z}^\top\bm{P} \bs{Z},
\]
where we denote $\bs{P}\defn\bs{\Sigma}^{1/2}S_k(S_k^\top \boldsymbol{\Sigma} S_k)^{-1} S_k^\top \boldsymbol{\Sigma}^{1/2}$.
To apply the second statement of Lemma~\ref{lem:gaussian-concentration}, we first calculate $\text{tr}(\bm{P})$ and $\|\bm{P}\|_F$.
By $\text{tr}(\bs{AB}) = \text{tr}(\bs{BA})$, it follows that $\text{tr}
(\bm{P}) = \text{tr}((S_k^\top \boldsymbol{\Sigma} S_k)^{-1}(S_k^\top \boldsymbol{\Sigma} S_k)) = \text{tr}(\bs{I}_k) = k$. Also notice that $\bs{P}$ is a projection matrix with rank $k$, and then $\|\bs{P}\|_F = \sqrt{\text{tr}(\bs{P}^\top\bs{P})} = \sqrt{\text{tr}(\bs{P})} = \sqrt{k}$. 
By choosing $t=\frac{3-2\sqrt{2}}{8} k$, we have, for some universal constant $C>0$,
\[
\mathbb{P}\left(
\bs{Z}^\top \bm{P} \bs{Z} \leq \frac{k}{\sqrt{2}}
\right) \leq \exp(-Ck).
\]
By the law of large numbers, $\|\bs{Z}\|_2^2/p \goto 1$ almost surely as $p\goto\infty$. Thus $\mathbb{P}(\|\bs{Z}\|_2^2 \geq \sqrt{2}p) \goto 0$. By the above reasoning and the following lower bound
\[
\mathbb{P}\left( \mathcal{E}_1\right)  \geq 1 - \mathbb{P}\left( \|\bs{Z}\|_2^2\geq \sqrt{2}p\right)  + \mathbb{P}\left( \bs{Z}^\top \bm{P} \bs{Z} \leq \frac{k}{\sqrt{2}}\right),
\]
we know $\mathbb{P}(\mathcal{E}_1)\goto 1$ as $k\goto\infty$ (recall that we assume $p\geq n/2$ and $k\goto\infty$ as $n\goto\infty$).

We complete the proof by showing $\mathbb{P}(\mathcal{E}_2)\goto 1$. Similar to the proof in the first part, we may write 
\[
\frac{\|\bs{\Sigma\beta}\|^2}{\bs{\beta}^\top\bs{\Sigma}\bs{\beta}} = \frac{1}{\|\bs{Z}\|^2} \bs{Z}^\top \bs{\Sigma}\bs{Z}.
\]
Slightly modifying the first statement of Lemma~\ref{lem:gaussian-concentration} yields
\begin{align*}
&\mathbb{P}\left( 
\bs{Z}^\top \bs{\Sigma}\bs{Z} \geq \text{tr}(\bs{\Sigma}) + 2 \|\cov\|_F \sqrt{t_1} + 2 \|\cov\|_2 t_2\right) \\
\leq ~&\mathbb{P}\left( 
\bs{Z}^\top \bs{\Sigma}\bs{Z} \geq \text{tr}(\bs{\Sigma}) + 2 \|\cov\|_F \sqrt{\min(t_1, t_2)} + 2 \|\cov\|_2 \min(t_1, t_2)\right) \\
\leq ~& \exp(-\min(t_1, t_2)).
\end{align*}
Choose $\sqrt{t_1} = \frac{\text{tr}(\cov)}{24 \|\cov\|_F}$ and $t_2 = \frac{\text{tr}(\cov)}{24 \|\cov\|_2}$. By $\|\cov\|_F \geq \|\cov\|_2$, we know $\sqrt{t_1}\leq t_2$. By Lemma~\ref{lem:matrix-norm} and Condition~\eqref{eq:zc-condition}, we observe $\sqrt{t_1} \goto\infty$ as $p\goto\infty$. Then
\[
\mathbb{P}\left(\bs{Z}^\top \bs{\Sigma}\bs{Z} \geq\sqrt{2}\text{tr}(\cov)\right)  \goto 0, \quad p\goto\infty.
\]
Similar to the first part, we have
\[
\mathbb{P}(\mathcal{E}_2) \geq 1 - \mathbb{P}(\|\bs{Z}\|^2 \geq \sqrt{2}p) - \mathbb{P}\left(\bs{Z}^\top \bs{\Sigma}\bs{Z} \geq\sqrt{2}\text{tr}(\cov)\right).
\]
Recall that we have shown $\mathbb{P}(\|\bs{Z}\|^2 \geq \sqrt{2}p) \goto 0$, and thus it follows that $\mathbb{P}(\mathcal{E}_2)\goto 1$.

\subsection{Details of Example~\ref{exp-1} and Example~\ref{exp-2}}\label{section:first two examples}
With the recommended choice $k=\lfloor n/2 \rfloor$, expression \eqref{eq: upper bound} in the main text becomes 
\begin{align*}
\text{ARE}_n(\Psi^{ZC}_n;\Psi^S_n) \leq
8\frac{\text{tr}(\boldsymbol{\Sigma})}{\sqrt{\text{tr}(\boldsymbol{\Sigma^2})}}  \frac{1}{\sqrt{n}}.
\end{align*}
Hence in both examples, we only need to deal with the right hand side of the inequality. 

For Example~\ref{exp-1}, we have 
\[
\text{tr}(\boldsymbol{\Sigma^2}) \geq \lambda_1^2 + \dots + \lambda_s^2 \overset{(i)}{\geq} (\lambda_1 + \dots + \lambda_s)^2 / s \overset{(ii)}{\geq} (1-\epsilon)^2 \text{tr}^2(\boldsymbol{\Sigma}) /s,
\]
where step~(i) follows by Cauchy-Schwarz inequality and step~(ii) uses the condition $\lambda_1+ \dots + \lambda_s \geq (1-\epsilon) \cdot \text{tr}(\boldsymbol{\Sigma})$. This inequality further implies that 
\begin{align*}
\text{ARE}_n(\Psi^{ZC}_n;\Psi^{S}_n) \leq 8\frac{\text{tr}(\boldsymbol{\Sigma})}{\sqrt{\text{tr}(\boldsymbol{\Sigma^2})}}  \frac{1}{\sqrt{n}} \leq \frac{8\sqrt{s}}{(1-\epsilon)\sqrt{n}}.
\end{align*}

Now we can see that $\sqrt{s} \leq (1-\epsilon) t_n \sqrt{n} / 8$ yields $\text{ARE}_n(\Psi^{ZC}_n;\Psi^{S}_n) \leq t_n$. When $s\asymp \sqrt{n}$, we have $\frac{8\sqrt{s}}{(1-\epsilon)\sqrt{n}} \asymp n^{-1/4}$ and then $\text{ARE}_n(\Psi^{ZC}_n;\Psi^{S}_n) \lesssim n^{-1/4}$.

For Example~\ref{exp-2}, with $\lambda_i(\bs{\Sigma}) = ai^{-1/2}$, we know 
\[
\text{tr}(\boldsymbol{\Sigma^2}) \leq a^2(1+\ln p); \quad \text{tr}(\boldsymbol{\Sigma}) \geq 2a(\sqrt{p} - 1) \quad \Longrightarrow \quad \text{ARE}_n(\Psi^{ZC}_n;\Psi^{S}_n) \leq 8 \frac{2 \sqrt{\rho}}{\sqrt{\ln p + 1}}.
\]
Thus we show the claims in Example~\ref{exp-1} and Example~\ref{exp-2}.


\subsection{Proof of Theorem~\ref{thm:optimality}}\label{section:theorem 2}
We establish Theorem~\ref{thm:optimality} by first proving an information theoretic lower bound and then proving that our test achieves this lower bound. 
Recall that in the proof, we use the fact that when the sketching dimension is chosen as $k\asymp r$, we have $\Delta_k^2/\bs{\beta}^\top\cov\bs{\beta} \overset{p}{\goto} 0$ (Lemma~\ref{lem:minimax}).

\subsubsection{Lower bound}
We start with the lower bound that is based on standard Le Cam's framework. Our argument is particularly similar to that in \cite{carpentier2018minimax}. Without loss of generality, we assume $\sigma^2 = 1$. First, we define a new parameter class $B_r(\tau)$ as
\[
B_r(\tau) = \left\lbrace \bs{\beta}\in\mathbb{R}^p: \|\bs{\beta}\|_2 \geq \tau, \beta_i = 0 \text{ for } r+1\leq i\leq p \right\rbrace.
\]
By definition of $\Theta_r(\tau)$, we can easily see that for any $\bs{\beta} \in B_r(\tau)$ and $\bs{\Sigma}_0 = \text{diag}\left(\bs{1}_r, \bs{0}_{p-r}\right)$, it follows $(\bs{\beta}, \bs{\Sigma}_0) \in \Theta_r(\tau)$.
Then the minimax Type II error can be bounded by
\[
\inf_{\psi} \sup_{\bs{\beta}\in B_r(\tau)} \mathbb{P}_{\bs{\beta}, \bs{\Sigma}_0}(\psi = 0) \leq \mathcal{R}_r(\tau).
\]
Let $\mu$ be a probability measure on $B_r(\tau)$. Consider any family of probability measures $P_{\bs{\beta}}$ indexed by $\bs{\beta}\in B_r(\tau)$. Denote by $\mathbb{P}_{\mu}$ the mixture probability measure 
\[
\mathbb{P}_{\mu} = \int_{B_r(\tau)} P_{\bs{\beta}} \ \mu(d\bs{\beta}).
\]
Also let $\chi^2(P', P) = \int (dP'/dP)^2 dP - 1$ be the chi-square divergence between two probability measures $P'\ll P$. Then,
\begin{align*}
\alpha + \mathcal{R}_r(\tau) & \geq \inf_{\psi}\sup_{\bs{\beta}\in B_r(\tau)}\left\lbrace \mathbb{P}_0(\psi = 1) + \mathbb{P}_{\bs{\beta}, \bs{\Sigma}_0}(\psi = 0) \right\rbrace \\
& \geq 1 - \sqrt{\chi^2(\mathbb{P}_{\mu}, P_0)},
\end{align*}
in which the infimum is taken over all test functions based on $(\bs{X}, \bm{y})$. To show the lower bound, it suffices to show that, for $\tau=\tau(A, n) =  \frac{A r^{1/4}}{\sqrt{n}}$, we can find $\mu_{\tau}$ such that
\begin{equation}
\chi^2(\mathbb{P}_{\mu_{\tau}}, P_0) \leq 1 + o_A(1),
\end{equation}
where $o_A(1)$ tends to $0$ as $A\goto 0$.

Note that when $\bs{\Sigma} = \bs{\Sigma}_0$, data matrix $\bs{X}$ under the null and alternative model only differs in the first $r$ features. Thus the chi-square divergence is essentially the divergence between two $r$-dimensional distributions, which allows us to borrow techniques for linear regression with $\bs{\Sigma} = \bs{I}_r$. More specifically, we may apply the results in Section 7.1 of \cite{collier2017minimax} and observe that
\begin{equation}
\chi^2(\mathbb{P}_{\mu_{\tau}}, P_0) \leq \exp(A^2)
\end{equation}
for some properly chosen $\mu_{\tau}$. See Section 7.1 of \cite{collier2017minimax} or Section 4.4 of \cite{carpentier2018minimax} for more details.

\subsubsection{Upper bound}
We now turn to the upper bound. Recall that we always assume $\bs{\beta}^\top\cov\bs{\beta} = O(1)$, since the problem is trivial otherwise. 
In order to show the upper bound, following the definition (ii), it suffices to show, if we choose $\psi^S$ to be the sketched $F$-test in Algorithm~\ref{alg:example} associated with any fixed sequence of sketching matrix $\{S_k\}\in\mathcal{A}$, it holds that
$$r(\psi^S, \bs{\beta}_n, \bs{\Sigma}_n) = o_P(1), \quad \text{when} \quad \tau_n/\epsilon_n \goto \infty \text{ and } (\bs{\beta}, \cov) \in \Theta_r(\tau).$$
For $(\bs{\beta}, \cov) \in \Theta_r(\tau)$, by Chebyshev inequality, we have 
\begin{equation}\label{eq:upper-bound}
\mathbb{E}_{\bs{\beta}, \cov}\left[1 - \psi^S\right] = \mathbb{P}(F^S > q_{\alpha, k, n-k})
\leq \frac{\text{Var}_{\bs{\beta}, \cov}(F^S)}{(q_{\alpha, k, n-k}- \mathbb{E}_{\bs{\beta}, \cov}\left[F^S\right])^2}.
\end{equation}

We claim that the following inequalities hold, and leave their proofs to the end of this section:
\begin{align}\label{eq:order}
& \text{Var}_{\bs{\beta}, \cov}(F^S) \leq \frac{C}{r^2} \left[ \frac{r^2+\lambda^2}{n} + (r+\lambda)\right] \quad \text{and} \\[.5em]
& (q_{\alpha, k, n-k}- \mathbb{E}_{\bs{\beta}, \cov}\left[F^S\right])^2 \geq \frac{\lambda^2}{2 r^2}, \label{eq:order2}
\end{align}
for any fixed $S_k \in \mathcal{A}$. Here we define $\lambda := n\Delta_k^2/\nu^2$ which satisfies $\sqrt{r}/\lambda = o_P(1)$.
As a consequence of expression \eqref{eq:order}, we have
\[
\mathbb{E}_{\bs{\beta}, \cov}\left[1 - \psi^S\right] \leq C \frac{(r^2+\lambda^2)/n + (r+\lambda)}{\lambda^2} = o_P(1).
\]
This completes the proof.

\subsubsection*{Proof of inequalities \eqref{eq:order} and \eqref{eq:order2}}
We omit the subscript $\bs{\beta}$  and $\cov$ of $\text{Var}$ and $\mathbb{E}$ for short. Recall that we define $\boldsymbol{\beta}^S = (S_k^\top \boldsymbol{\Sigma} S_k)^{-1} S_k^\top \boldsymbol{\Sigma} \boldsymbol{\beta}$ and $\nu^2 = \sigma^2 + \boldsymbol{\beta}^\top \boldsymbol{\Sigma}\boldsymbol{\beta}-\Delta_k^2$. Following the reasoning in the proof of Theorem~\ref{thm:power-function},we have under $H_1$,  
\[
F^S \vert \bs{X} \sim F_{k,n-k}(\lambda(\bs{X})) \quad \text{where} \quad \lambda(\bs{X}) = \frac{(\bs{\beta}^S)^\top S_k^\top \bs{X}^\top \bs{X} S_k \bs{\beta}^S}{\nu^2}.
\]
By the moment expressions of a non-central $F$-statistic, it 
can be easily seen that
\begin{align*}
& \mathbb{E}[F^S|\bs{X}] = \frac{(n-k)(k+\lambda(\bs{X}))}{k(n-k-2)}, \\[.5em]
& \text{Var}(F^S|\bs{X}) = 2 \frac{(k+\lambda(\bs{X}))^2 + (k + 2\lambda(\bs{X}))(n-k-2)}{(n-k-2)^2(n-k-4)} \left(\frac{n-k}{k}\right)^2.
\end{align*}
Then we have, with $\lambda \defn \mathbb{E}[\lambda(\bs{X})] = n\Delta_k^2/\nu^2$ and $\text{Var}(\lambda(\bs{X})) = 2\lambda^2/n$,
\begin{align*}
& \text{Var}(\mathbb{E}[F^S|\bs{X}]) \leq \frac{2}{k^2} \text{Var}(\lambda(\bs{X})), \\[.5em]
& \mathbb{E}[\text{Var}(F^S|\bs{X})] \leq \frac{C}{k^2} \left[ \frac{(k+\lambda)^2}{n} + (k+\lambda) + \frac{\text{Var}(\lambda(\bs{X}))}{n}
\right].
\end{align*}
By the law of total variance,
\[
\text{Var}(F^S) = \text{Var}(\mathbb{E}[F^S|\bs{X}]) + \mathbb{E}[\text{Var}(F^S|\bs{X})] \leq \frac{C}{k^2} \left[ \frac{(k+\lambda)^2}{n} + (k+\lambda)\right],
\]
which proves inequality~\eqref{eq:order} under the assumption $k \asymp r$. 

To prove inequality~\eqref{eq:order2}, notice that
\[
\mathbb{E}[F^S] = \frac{(n-k)(k+\lambda)}{k(n-k-2)}.
\]
In addition, Lemma C.8. of \cite{steinberger2016relative} yields
\[
q_{\alpha, k, n-k} = 1 + \sqrt{\frac{2n}{k(n-k)}} z_{\alpha} + o(k^{-1/2}).
\]
By the assumption $\tau_n/\epsilon_n \goto \infty$, it follows $\lambda \gg \sqrt{r}$. After checking each term in $(q_{\alpha, k, n-k}-\mathbb{E}[F^S])^2$, we have 
\[
(q_{\alpha, k, n-k}-\mathbb{E}[F^S])^2 = \frac{\lambda^2}{k^2}(1+o(1)),
\]
which verifies inequality~\eqref{eq:order2}.


\subsection{Proof of Theorem~\ref{thm:power-3}}\label{section:theorem on power guarantee}
By Theorem~\ref{thm:power-2} and Lemma~\ref{lem:minimax}, it suffices to show that when $\Delta_k^2/ \bs{\beta}^\top\cov\bs{\beta} \overset{p}{\goto} 1$, we have 
\begin{equation}
y_n := \Phi\left(
- z_{\alpha} + \frac{\sqrt{n}\Delta_k^2}{\sigma^2} \sqrt{\frac{1-k/n}{2k/n}}
\right) - \Phi\left(
- z_{\alpha} + \frac{\sqrt{n}\bs{\beta}^\top\cov\bs{\beta}}{\sigma^2} \sqrt{\frac{1-k/n}{2k/n}}
\right) \overset{p}{\goto} 0.
\end{equation}
For ease of notation, let us write $a_n = \frac{\sqrt{n}\bs{\beta}^\top\cov\bs{\beta}}{\sigma^2} \sqrt{\frac{1-k/n}{2k/n}}$ and $\eta_n = \frac{\bs{\beta}^\top\cov\bs{\beta}-\Delta_k^2}{\bs{\beta}^\top\cov\bs{\beta}}$. Then we have $\eta_n\overset{p}{\goto} 0$ and $\eta_n\geq 0$, due to the fact that $\Delta_k^2\leq \bm{\beta}^\top\cov\bm{\beta}$. Assume $n$ is large enough,  such that $ \eta_n \leq 1/2$.
By Lipschitz-1 property of $\Phi(\cdot)$, we have 
\begin{equation}\label{eq:yn-1}
|y_n| \leq \eta_n a_n.
\end{equation}
On the other hand, we have
\begin{equation}\label{eq:yn-2}
\begin{aligned}
|y_n| & \overset{\text{(i)}}{\leq} \Phi\left(
z_{\alpha} - \frac{\sqrt{n}\Delta_k^2}{\sigma^2} \sqrt{\frac{1-k/n}{2k/n}}
\right) + \Phi\left(z_{\alpha} -  \frac{\sqrt{n}\bs{\beta}^\top\cov\bs{\beta}}{\sigma^2} \sqrt{\frac{1-k/n}{2k/n}}\right) \\
& \leq 2 \Phi\left(
z_{\alpha} - \frac{\sqrt{n}\Delta_k^2}{\sigma^2} \sqrt{\frac{1-k/n}{2k/n}}
\right) \\
& \overset{\text{(ii)}}{\leq} 2 \Phi(z_{\alpha} - a_n/2)\\
& \overset{\text{(iii)}}{\leq} 2 \exp\left\lbrace - \frac{\bm{1}_{\left\lbrace z_{\alpha} - a_n/2\leq 0\right\rbrace } (z_{\alpha} - a_n/2)^2 }{2}\right\rbrace,
\end{aligned}
\end{equation}
where step (i) is due to $\Phi(x)-\Phi(y) \leq \left| \Phi(x)-\Phi(y)\right| = \left| \Phi(-x)-\Phi(-y)\right|\leq \Phi(-x)+\Phi(-y)$,
step (ii) follows from $\eta_n\leq 1/2$ and step (iii) uses the Gaussian tail bound $\Phi(x) \leq 2\exp(-\bm{1}_{\left\lbrace x\leq 0\right\rbrace } x^2/2)$.

Combining inequalities \eqref{eq:yn-1} and \eqref{eq:yn-2}, we have 
\[
|y_n| \leq \min\left\lbrace \eta_n a_n, 2 \exp\left\lbrace - \frac{\bm{1}_{\left\lbrace z_{\alpha} - a_n/2\leq 0\right\rbrace } (z_{\alpha} - a_n/2)^2 }{2}\right\rbrace \right\rbrace.
\]
Given $\eta_n>0$, we know $\eta_n a_n$ and $2 \exp\left\lbrace - \bm{1}_{\left\lbrace z_{\alpha} - a_n/2\leq 0\right\rbrace } (z_{\alpha} - a_n/2)^2 /2\right\rbrace $ are monotone increasing and decreasing respectively as functions of $a_n$. Then we have the upper bound
\begin{equation}\label{eq:bound-of-yn}
|y_n| \leq 2 \exp\left\lbrace - \frac{\bm{1}_{\left\lbrace z_{\alpha} - f(\eta_n)/2\leq 0\right\rbrace } (z_{\alpha} - f(\eta_n)/2)^2 }{2}\right\rbrace,
\end{equation}
where $f(\eta_n)$ is the unique $x_n$ that solves
\[
\eta_n x_n = 2 \exp\left\lbrace - \frac{\bm{1}_{\left\lbrace z_{\alpha} - x_n/2\leq 0\right\rbrace } (z_{\alpha} - x_n/2)^2 }{2}\right\rbrace .
\]
We can directly check that $f(\eta_n)$ is a monotone decreasing function of $\eta_n$, and $\displaystyle{\lim_{\eta_n\rightarrow 0^+}} f(\eta_n) = +\infty$. Then 
\[
\displaystyle{\lim_{\eta_n\rightarrow 0^+}} 2 \exp\left\lbrace - \frac{\bm{1}_{\left\lbrace z_{\alpha} - f(\eta_n)/2\leq 0\right\rbrace } (z_{\alpha} - f(\eta_n)/2)^2 }{2}\right\rbrace= 0.
\]
By bound~\eqref{eq:bound-of-yn}, it follows that $y_n \overset{p}{\goto} 0$. This completes the proof of Theorem~\ref{thm:power-3}.


\subsection{Details of Examples in Section~\ref{sec:examples}}
\label{section:details of second set of examples}
In this section, we provide details of Examples~\ref{example:polynomial}--\ref{example:beta-structure} with $\eta = 1/\log p$. Note that for Examples~\ref{example:polynomial}--\ref{example:block}, the conditions in Definition~\ref{def: intrinstic dim} essentially boil down to $	\sum_{i=r+1}^p \lambda_i \leq \eta \sum_{i=1}^p \lambda_i $ and $r\lambda_{r+1} \leq \eta \sum_{i=1}^p \lambda_i$.

To start with Example~\ref{example:polynomial}, notice that
\[
\sum_{i=r+1}^p\lambda_i \asymp r^{-\alpha + 1}, \quad  \sum_{i=1}^p \lambda_i \asymp 1.
\]
Then the conditions translate to $r^{-\alpha + 1} \leq 1/\log p$, or equivalently, there exists $r \lesssim (\log p)^{\frac{1}{\alpha-1}}$ such that $(\bm{\beta}, \cov)\in \mathcal{D}(r)$ as stated in Example~\ref{example:polynomial}.

For Example~\ref{example:exponential}, first note that $\lambda_k/\lambda_{k+1} = \exp((k+1)^{\gamma}-k^{\gamma})$ by definition. When $\gamma\geq 1$, we have $\lambda_k/\lambda_{k+1}\geq e\geq 1 + 1/k$; whereas $0 < \gamma < 1$, we have
$$(k+1)^{\gamma}-k^{\gamma} = \gamma \int_{k}^{k+1} x^{\gamma-1} dx \geq \gamma (k+1)^{\gamma-1} \geq \frac{\gamma}{k}.$$
Thus $\lambda_k/\lambda_{k+1} \geq 1 + \gamma/k$. In either case, we know $\left\lbrace \lambda_k\right\rbrace $ decays faster than $(\gamma \land 1)/k$. Then observe that
\[
\sum_{i=r+1}^p\lambda_i \leq \lambda_r \sum_{i=r+1}^{p} \frac{\gamma \land 1}{i} \lesssim \exp(-r^{\gamma}) (\log p), \quad  \sum_{i=1}^p \lambda_i \asymp 1.
\]
Thus the conditions translate to $\exp(-r^{\gamma}) (\log p) \leq 1/\log p$ and $r\exp(-r^{\gamma})\leq 1/\log p$, or there exists $r \lesssim (\log \log p)^{\frac{1}{\gamma}}$ such that $(\bm{\beta}, \cov)\in \mathcal{D}(r)$ as stated in Example~~\ref{example:exponential}.

For Example~\ref{example:block}, when $r\geq m$, we have 
\begin{align*}
\sum_{i=r+1}^p \lambda_i \leq p(1-\rho) \lesssim \frac{p}{\log p} \quad \text{and} \quad \sum_{i=1}^p \lambda_i = p.
\end{align*}
Then we can directly see the two conditions are satisfied.

For Example~\ref{example:beta-structure}, we know that 
$$\sum_{i=1}^r \widetilde{\beta}_i \asymp \log r, \quad \sum_{i=r+1}^p \lambda_i \asymp \log p - \log r, \quad \sum_{i=r+1}^p \widetilde{\beta}_i^2 \lambda_i\asymp 1/r, \quad \bs{\beta}^\top \cov \bs{\beta} \asymp 1.$$
Then the first condition of Definition~\ref{def: intrinstic dim} is now $\log r (\log p - \log r)/r + 1/r \leq 1/\log p$, or $r \geq \log^2 p \log r$. Thus we can see that there exists $r \lesssim (\log p)^3$ satisfying both conditions.


\section{Auxiliary proofs}
\label{section:additional-proof}

\subsection{Remaining proof of Lemma~\ref{lem:minimax}}
In this part, we prove that, when $(\bs{\beta}, \cov)\in\mathcal{D}(r)$ and $ar \leq k\leq br$, we have $\Delta_k^2/\bs{\beta}^\top\cov\bs{\beta} \overset{p}{\goto } 1$. 
In order to show this convergence result, we make use of the following lemma. 
\begin{lemma}\label{lem:spectral-bounds}
	For $\forall a>1$, if we choose sketching dimension to be $a r \leq k\leq C_1 \ \frac{\sum_{i=r+1}^p \lambda_i}{\lambda_{r+1}}$, then with probability at least $1-\exp(-c_2r)-\exp(-c_1\frac{\sum_{i=r+1}^p \lambda_i}{\lambda_{r+1}})$, we obtain 
	\begin{enumerate}
		\item $\kappa(\widetilde{S}_2^\top \Lambda_{p-r}\widetilde{S}_2) \leq 4$; 
		\item $\lambda_{\max}(\widetilde{S}_2^\top \Lambda_{p-r}\widetilde{S}_2) \leq 2 \sum_{i=r+1}^p \lambda_i$;
		\item $\kappa(\widetilde{S}_1\widetilde{S}_1^\top) \leq C_2$;
		\item $\lambda_{\min}(\widetilde{S}_1\widetilde{S}_1^\top) \geq C_2^{-1} k$,
	\end{enumerate}
	where $c_1, c_2, C_1, C_2$ are universal constants only depending on $a$.
\end{lemma}

The proof of Lemma~\ref{lem:spectral-bounds} can be found at the end of this section. Now suppose that Lemma~\ref{lem:spectral-bounds} is given and also assume that $ar \leq k \leq br$ with $a>1$.
Recall from the main text that for any $\bs{\xi}$ that satisfies condition \eqref{eq:restriction}, i.e.
$\bs{U}_r^\top\bs{\beta} = \bs{U}_r^\top S_k\bs{\xi}$, it holds that
\begin{equation*}
\Delta_k^2  \geq  \bs{\beta}^\top \cov \bs{\beta} -(\bs{\beta}-S_k\bs{\xi})^\top \cov_{p-r} (\bs{\beta}-S_k\bs{\xi}).
\end{equation*}
Given any $\cov_{p-r}^+ \succeq \cov_{p-r}$, it is easily seen that the following inequality always holds 
\begin{equation*}
\Delta_k^2  \geq  \bs{\beta}^\top \cov \bs{\beta} -(\bs{\beta}-S_k\bs{\xi})^\top \cov_{p-r}^+ (\bs{\beta}-S_k\bs{\xi}).
\end{equation*}
With this new $\cov_{p-r}^+$ and under the same restriction $\bs{U}_r^\top\bs{\beta} = \bs{U}_r^\top S_k\bs{\xi}$, we can find $\bs{\xi}^*$ that minimizes
$(\bs{\beta}-S_k\bs{\xi})^\top \cov_{p-r}^+ (\bs{\beta}-S_k\bs{\xi})$
and then apply Lemma~\ref{lem:delta-bound-1} to obtain
\[
\bs{\beta}^\top\cov\bs{\beta} - \Delta_k^2 \leq \frac{2\|\widetilde{S}_2^\top\bs{\Lambda}_{p-r}^+ \widetilde{S}_2\|_2}{\lambda_{\min}(\widetilde{S}_1\widetilde{S}_1^\top)}\|\widetilde{\bs{\beta}}_1\|_2^2 + 2\left( 1 + \kappa(\widetilde{S}_2^\top\bs{\Lambda}_{p-r}^+ \widetilde{S}_2) \kappa(\widetilde{S}_1\widetilde{S}_1^\top) \right) \cdot \widetilde{\bs{\beta}}_2^\top \bs{\Lambda}_{p-r}^+ \widetilde{\bs{\beta}}_2.
\]
Here we write $\cov_{p-r}^+ = \bs{U}_{p-r}^\top \bs{\Lambda}_{p-r}^+\bs{U}_{p-r}$, with $\bs{\Lambda}_{p-r}^+ = \text{diag}(\lambda_{r+1}^+, \dots, \lambda_p^+)$. Then by Lemma~\ref{lem:minimal-eigenvalue}, when sketching dimension $k$ satisfies $ar \leq k \leq C_1 \sum_{i=r+1}^p \lambda_i^+ / \lambda_{r+1}^+$, we have
\begin{equation}\label{eq:delta-bound-lifting}
\begin{aligned}
\bs{\beta}^\top\cov\bs{\beta} - \Delta_k^2 ~\leq~ & 
\frac{4C_2\sum_{i=r+1}^p \lambda_i^+}{k} \|\widetilde{\bs{\beta}}_1\|_2^2 + 2(1+4C_2) \cdot \widetilde{\bs{\beta}}_2^\top \bs{\Lambda}_{p-r}^+ \widetilde{\bs{\beta}}_2 \\[.5em]
~\leq~  & C_3 \left(
\left(\frac{1}{r}\sum_{i=1}^r \widetilde{\beta}_i^2 \right) \cdot \left(\sum_{i=r+1}^p \lambda_i^+ \right) + \sum_{i=r+1}^p \widetilde{\beta}_i^2 \lambda_i^+
\right)
\end{aligned}
\end{equation}
with probability at least $1-\exp(-c_2r)-\exp(-c_1 \sum_{i=r+1}^p \lambda_i^+/\lambda_{r+1}^+)$. Note that the constant $C_3$ here only depends on $a$.

Up to now, the derivations do not depend on the form of matrix $\cov_{p-r}^+$. 
Now we are ready to choose a particular form of $\cov_{p-r}^+$, namely we can set 
\begin{equation}\label{eq:eigen-lifting}
\cov_{p-r}^+ \defn \bs{U}_{p-r}^\top \bs{\Lambda}_{p-r}^+\bs{U}_{p-r} \quad \text{and} \quad \lambda_i^+ = \lambda_i + \frac{b}{C_1} \frac{r\lambda_{r+1}}{p-r}, \quad \text{for $r+1 \leq i\leq p$.}
\end{equation}
Then direct calculations give 
$$\frac{\sum_{i=r+1}^p \lambda_i^+}{\lambda_{r+1}^+} = 
\frac{\sum_{i=r+1}^p \lambda_i + \frac{b}{C_1} r\lambda_{r+1}}{\lambda_{r+1} + \frac{b}{C_1} \frac{r\lambda_{r+1}}{p-r}} \geq \frac{b}{C_1} r.
$$
Plugging the expression of $\lambda_i^+$ into \eqref{eq:delta-bound-lifting} and then applying the conditions in Definition~\ref{def: intrinstic dim} yield the following result:

When sketching dimension $k$ satisfies $ar\leq k\leq br$, we have with probability at least $1-2\exp(-c_3r)$ that
\[
1 - \frac{\Delta_k^2}{\bs{\beta}^\top\cov\bs{\beta}} \leq C_3\left(1+\frac{b}{C_1}\right) \eta = o(1).
\]
Thus we finish the proof with the stronger conclusion $\frac{\Delta_k^2}{\bs{\beta}^\top\cov\bs{\beta}}\overset{p}{\goto} 1$.
Now we are only left to prove Lemma~\ref{lem:spectral-bounds}.

\paragraph*{Proof of Lemma~\ref{lem:spectral-bounds}.}
To establish Lemma~\ref{lem:spectral-bounds}, we make use of the result below whose proof is provided in Section~\ref{SecPfMatrix-norm}.
\begin{lemma}\label{lem:minimal-eigenvalue}
	Suppose $\bs{\Lambda} = \text{\emph{diag}}(\lambda_1, \dots, \lambda_N)$ with $\lambda_i\geq0$, $\|\lambda\|^2 >0$ and $S\in\mathbb{R}^{N\times n}$ is a standard Gaussian random matrix with $n\leq N$. Write $\bs{\lambda} = (\lambda_1, \dots, \lambda_N)$. Then for $t<1$, 
	\[
	(1 - t)\sqrt{\sum_{i=1}^N \lambda_i^2} \leq s_{\min}(\bs{\Lambda S}) \leq s_{\max}(\bs{\Lambda S}) \leq (1 + t)\sqrt{\sum_{i=1}^N \lambda_i^2},
	\]
	with probability at least $1-9^n \cdot 2\exp \left(-\min\left\lbrace
	\frac{1}{16} \frac{\|\bs{\lambda}\|_2^4}{\|\bs{\lambda}\|^4_4} t^2 , \frac{1}{4} \frac{\|\bs{\lambda}\|_2^2}{\|\bs{\lambda}\|_{\infty}^2} t \right\rbrace\right)$.
\end{lemma}
Applying Lemma~\ref{lem:minimal-eigenvalue} to $\bs{\Lambda}_{p-r}^{1/2} \widetilde{S}_2$ with sketching dimension $k$ and $t=1/3$ yields
\begin{equation}\label{eq:delta-bound-3}
\kappa(\widetilde{S}_2^\top \bs{\Lambda}_{p-r}\widetilde{S}_2) \leq 4 \quad \text{and} \quad \lambda_{\max}(\widetilde{S}_2^\top \bs{\Lambda}_{p-r}\widetilde{S}_2) \leq 2\sum_{i=r+1}^p \lambda_i
\end{equation}
with probability at least $1-\exp\left(\ln 9 \cdot k - \min\left\lbrace \frac{1}{144} \frac{(\sum_{i=r+1}^p \lambda_i)^2}{\sum_{i=r+1}^p \lambda_i^2}, \frac{1}{12} \frac{\sum_{i=r+1}^p \lambda_i}{\lambda_{r+1}} \right\rbrace\right)$. Since $\frac{(\sum_{i=r+1}^p \lambda_i)^2}{\sum_{i=r+1}^p \lambda_i^2} \geq \frac{\sum_{i=r+1}^p \lambda_i}{\lambda_{r+1}}$, \eqref{eq:delta-bound-3} holds with probability at least $1 - \exp\left(- c_1 \frac{\sum_{i=r+1}^p \lambda_i}{\lambda_{r+1}} \right)$  as long as $k\leq C_1 \ \frac{\sum_{i=r+1}^p \lambda_i}{\lambda_{r+1}}$. 

Lemma~\ref{lem:minimal-eigenvalue} can be used to bound all the four quantities in Lemma~\ref{lem:spectral-bounds}. To obtain a better control for $\kappa(\widetilde{S}_1\widetilde{S}_1^\top)$ and $\lambda_{\min}(\widetilde{S}_1\widetilde{S}_1^\top)$ in terms of constants, we invoke the following lemma from \cite{NIPS2011_4260}:
\begin{lemma}[Lemma 4 of \cite{NIPS2011_4260}] \label{lem:minimal-eigenvalue-2}
	For $k\leq p$, let $P_k \in \mathbb{R}^{k\times p}$ be a random matrix with i.i.d. $\mathcal{N}(0, 1)$ entries. Then
	\begin{align*}
	\mathbb{P}\left(\lambda_{\max}(\frac{1}{p} P_k^\top P_k) \geq (1+\sqrt{k/p} + t)^2 \right) \leq \exp(-pt^2/2);\\
	\mathbb{P}\left(\lambda_{\min}(\frac{1}{p} P_k^\top P_k) \leq (1-\sqrt{k/p} - t)^2 \right) \leq \exp(-pt^2/2).
	\end{align*}
\end{lemma}

With constant $a>1$ and $k\geq ar$, we now apply Lemma~\ref{lem:minimal-eigenvalue-2} to $\widetilde{S}_1$ and obtain that with probability at least $1-\exp(-c_2 r)$, 
\begin{equation}\label{delta-bound-4}
\kappa(\widetilde{S}_1 \widetilde{S}_1^\top )\leq C_2 \quad \text{and} \quad \lambda_{\min}(\widetilde{S}_1 \widetilde{S}_1^\top) \geq C_2^{-1} k
\end{equation}
where $c_2, C_2$ are universal constants only depending on $a$. This completes the proof of Lemma~\ref{lem:spectral-bounds}.


\subsection{Proof of Lemma~\ref{lem:delta-bound-1}}\label{section:lemma on eigenvalues}
\textbf{Structure of the proof:} We prove Lemma~\ref{lem:delta-bound-1} following the Lagrange multiplier procedure discussed in the main text. We first derive the expression of $\bs{\xi}^*$ using the Lagrange multiplier; the explicit form of $\bs{\xi}^*$ is summarized in \eqref{eq:xi-optimal} and \eqref{eq:lambda-optimal}. Then we plug $\bs{\xi}^*$ into $(\bm{\beta}-S_k\bm{\xi}^*)^\top \cov_{p-r}(\bm{\beta}-S_k\bm{\xi}^*)$, and get its upper bound; see \eqref{eq:delta-lower-bound-1}. The remaining part of the proof proceeds by bounding the terms in \eqref{eq:delta-lower-bound-1} based on properties of the spectral norm. 

\subsubsection*{Step 1: Finding minimal value of $(\bm{\beta}-S_k\bm{\xi}^*)^\top \cov_{p-r}(\bm{\beta}-S_k\bm{\xi}^*)$.}

Recall that we define the Lagrange form
$$
\mathcal{L}(\bs{\xi}, \bs{\lambda}) = \frac{1}{2}(\bs{\beta}-S_k\bs{\xi})^\top \cov_{p-r}(\bs{\beta}-S_k\bs{\xi}) - \bs{\lambda}^\top \bs{U}_r^\top (\bs{\beta}-S_k\bs{\xi}).
$$
By solving the following two equations
\begin{equation*}\label{eq:xi-star}
\frac{\partial \mathcal{L}(\bs{\xi}, \bs{\lambda})}{\partial \bs{\xi}} = 0  \quad \text{and} \quad
\bs{U}_r^\top (\bs{\beta}-S_k\bs{\xi})=0,
\end{equation*}
we can obtain the optimal solution $\bs{\xi}^*$.

First, let us consider the first equation
\[
\frac{\partial \mathcal{L}(\bm{\xi},\bs{\lambda})}{\partial \bs{\xi}} = 0.
\]
A direct calculation yields
\[
S_k^\top \cov_{p-r} S_k\bm{\xi} - S_k^\top \cov_{p-r}\bs{\beta} + S_k^\top \bm{U}_r \bs{\lambda} = 0.
\]
Similar to proof in Section~\ref{Section: the proof of invertibility} and by noting that $\text{rank}(\cov_{p-r}) = p-r \geq k$, we can show the matrix $S_k^\top \cov_{p-r} S_k$ is invertible almost surely. Then the solution can be written explicitly as 
\begin{equation}\label{eq:xi-optimal}
\bm{\xi}^* = (S_k^\top \cov_{p-r} S_k)^{-1}S_k^\top (\cov_{p-r}\bs{\beta} - \bm{U}_r \bs{\lambda}^*).
\end{equation}
By writing $\bs{H} = S_k (S_k^\top \cov_{p-r} S_k)^{-1}S_k^\top$ and plugging the above expression to the constraint condition $\bs{U}_r^\top (\bs{\beta}-S_k\bs{\xi})=0$, we obtain the following equality:
\[
\bm{U}_r^\top \bs{\beta} - \bm{U}_r^\top \bs{H} (\cov_{p-r}\bs{\beta} - \bm{U}_r\bs{ \lambda}^*) = 0.
\]
Before preceding, we first justify that $\bm{U}_r^\top \bm{H} \bm{U}_r$ is invertible almost surely. Note that when $S_k^\top \cov_{p-r} S_k$ is invertible, we have $\bm{x}^\top \bm{U}_r^\top \bs{H}\bm{U}_r \bm{x}=0$ iff $S_k^\top \bm{U}_r\bm{x}=0$ iff $\bm{x}^\top \bm{U}_r^\top S_k S_k^\top \bm{U}_r \bm{x}=0$. Since $S_k^\top \bm{U}_r \in \mathbb{R}^{k\times r}$ is distributed as an i.i.d Gaussian sketching matrix, we conclude that $\text{rank}(\bm{U}_r^\top S_k S_k^\top \bm{U}_r )=\text{rank}(S_k^\top \bm{U}_r) = r$ almost surely with $k\geq r$. Now with $S_k^\top \cov_{p-r} S_k$ invertible and $\text{rank}(\bm{U}_r^\top S_k S_k^\top \bm{U}_r )= r$ (which happens almost surely), we know that $\bm{x}^\top \bm{U}_r^\top \bs{H}\bm{U}_r \bm{x}=0$ iff $\bm{x}=0$, or equivalently, $\bm{U}_r^\top \bm{H} \bm{U}_r$ is invertible.

Now we can safely write $(\bm{U}_r^\top \bs{H}\bm{U}_r)^{-1}$. In this case,
\begin{equation}\label{eq:lambda-optimal}
\bs{\lambda}^* = (\bm{U}_r^\top \bs{H}\bm{U}_r)^{-1}(\bm{U}_r^\top \bs{H}\cov_{p-r}-\bm{U}_r^\top)\bs{\beta}.
\end{equation}

Based on \eqref{eq:xi-optimal} and \eqref{eq:lambda-optimal}, we have 
\begin{align*}
S_k\bm{\xi}^* &= \bm{H}\left(\cov_{p-r}\bs{\beta}-\bm{U}_r (\bm{U}_r^\top \bm{H}\bm{U}_r)^{-1}(\bm{U}_r^\top \bm{H} \cov_{p-r}-\bm{U}_r^\top)\bs{\beta}\right) \\
& = \bm{H}\left( \cov_{p-r} - \bm{U}_r (\bm{U}_r^\top \bs{H}\bm{U}_r)^{-1}\bm{U}_r^\top (\bm{H}\cov_{p-r}-\Ind) \right) \bs{\beta}.	
\end{align*}
With the above expression at hand, we are ready to control quantity 
$(\bm{\beta}-S_k\bm{\xi}^*)^\top \cov_{p-r}(\bm{\beta}-S_k\bm{\xi}^*)$.
For the sake of notational simplicity, let us write $$\bm{G} \defn \bm{U}_r (\bm{U}_r^\top \bm{H}\bm{U}_r)^{-1}\bm{U}_r^\top.$$ Then we obtain the following equalities:
\begin{align*}
S_k\bm{\xi}^* &=\bm{H}\left( \cov_{p-r} - \bm{G}(\bm{H}\cov_{p-r}-\Ind) \right) \bm{\beta}; \\
\bm{\beta}-S_k\bm{\xi}^* 
&= (\Ind-\bm{HG})(\Ind-\bm{H}\cov_{p-r})\bm{\beta}.
\end{align*}
Putting the pieces together yields
\[
(\bm{\beta}-S_k\bm{\xi}^*)^\top \cov_{p-r}(\bm{\beta}-S_k\bm{\xi}^*) = \bm{\beta}^\top (\Ind-\cov_{p-r} \bm{H})(\Ind-\bm{GH})\cov_{p-r}(\Ind-\bm{HG})(\Ind-\bm{H}\cov_{p-r})\bm{\beta}.
\]

\subsubsection*{Step 2: Upper bounding the minimal value.}

By recalling the notation $\bs{H}=S_k(S_k^\top \cov_{p-r} S_k)^{-1}S_k^\top$, we know $\bs{H}\cov_{p-r} \bs{H}=\bs{H}$ and $\bs{G}\bs{H}\bs{G}=\bs{G}$. 
Then $\bm{H}\cov_{p-r}(\Ind-\bs{H}\cov_{p-r})=\bs{0}$, and 
\begin{align*}
(\bs{\beta}-S_k\bm{\xi}^*)^\top \cov_{p-r}(\bs{\beta}-S_k\bm{\xi}^*) 
=&\ \bs{\beta}^\top (\Ind-\cov_{p-r} \bs{H})(\cov_{p-r} - \bs{G}\bs{H}\cov_{p-r}-\cov_{p-r}\bs{H}\bs{G}+\bs{G})(\Ind-\bs{H}\cov_{p-r})\bs{\beta} \\[.5em]
=&\ \bs{\beta}^\top (\Ind-\cov_{p-r} \bs{H})(\cov_{p-r} + \bs{G})(\Ind-\bs{H}\cov_{p-r})\bs{\beta} \\[.5em]
=&\ \bs{\beta}^\top\cov_{p-r}\bs{\beta}-\bs{\beta}^\top\cov_{p-r}\bs{H}\cov_{p-r}\bs{\beta} + \bs{\beta}^\top (\Ind-\cov_{p-r} \bs{H}) \bs{G} (\Ind-\bs{H}\cov_{p-r})\bs{\beta} \\[.5em]
\leq & \ \bs{\beta}^\top\cov_{p-r}\bs{\beta} + \|(\bm{U}_r^\top \bs{H} \bm{U}_r)^{-1}\|_2 \|\bm{U}_r^\top \bs{\beta} - \bm{U}_r^\top \bs{H} \cov_{p-r}\bs{\beta}\|_2^2.
\end{align*}
By definition of $\widetilde{S}_1 \defn \bm{U}_r^\top S_k$ and $\widetilde{S}_2 \defn \bs{U}_{p-r}^\top S_k$, we can see $\widetilde{S}_1$ and $\widetilde{S}_2$ are independent, and their entries are independent standard Gaussian random variables.
Additionally denoting $\bs{\widetilde{\beta}}_1 \defn \bm{U}_r^\top\bs{\beta}$ and $\bs{\widetilde{\beta}}_2 \defn \bs{U}_{p-r}^\top\bs{\beta}$, we can rewrite the above as 
\begin{equation}\label{eq:delta-lower-bound-1}
(\bs{\beta}-S_k\bm{\xi}^*)^\top \cov_{p-r}(\bs{\beta}-S_k\bm{\xi}^*) 
\leq\widetilde{\bs{\beta}}_2^\top \bs{\Lambda}_{p-r} \widetilde{\bs{\beta}}_2 + \|(\bm{U}_r^\top \bs{H} \bm{U}_r)^{-1}\|_2 \|\widetilde{\bs{\beta}}_1 - \widetilde{S}_1 (\widetilde{S}_2^\top \bs{\Lambda}_{p-r} \widetilde{S}_2)^{-1} \widetilde{S}_2^\top \bs{\Lambda}_{p-r}\widetilde{\bs{\beta}}_2\|_2^2.
\end{equation}
With some algebra (see the details at the end of this section), it can be shown that 
\begin{align}\label{eq:delta-lower-bound-2}
\|\widetilde{\bs{\beta}}_1 - \widetilde{S}_1 (\widetilde{S}_2^\top \bs{\Lambda}_{p-r} \widetilde{S}_2)^{-1} \widetilde{S}_2^\top \bs{\Lambda}_{p-r}\widetilde{\bs{\beta}}_2\|_2^2 & \leq 2 \|\widetilde{\bs{\beta}}_1 \|_2^2 + 2 \|\widetilde{S}_1\|_2^2 \|(\widetilde{S}_2^\top \bs{\Lambda}_{p-r} \widetilde{S}_2)^{-1}\|_2;\\
\label{eq:delta-lower-bound-3}
\|(\bm{U}_r^\top \bs{H} \bm{U}_r)^{-1}\|_2
& \leq \frac{\lambda_{\max}(\widetilde{S}_2^\top \bs{\Lambda}_{p-r} \widetilde{S}_2)}{\lambda_{\min}(\widetilde{S}_1\widetilde{S}_1^\top)}.
\end{align}

Plugging inequalities \eqref{eq:delta-lower-bound-2} and \eqref{eq:delta-lower-bound-3} into \eqref{eq:delta-lower-bound-1} yields
\begin{align*}
&\ (\bs{\beta}-S_k\bm{\xi}^*)^\top \cov_{p-r}(\bs{\beta}-S_k\bm{\xi}^*) \\
\overset{\text{by }\eqref{eq:delta-lower-bound-2}}{\leq} & \ 2 \|(\bm{U}_r^\top \bs{H} \bm{U}_r)^{-1}\|_2 \|\widetilde{\bs{\beta}}_1\|_2^2 + \left(1 + 2 \|(\bm{U}_r^\top \bs{H} \bm{U}_r)^{-1}\|_2 \|\widetilde{S}_1\|_2^2 \| (\widetilde{S}_2^\top \bs{\Lambda}_{p-r} \widetilde{S}_2)^{-1}\|_2 \right)\cdot \widetilde{\bs{\beta}}_2^\top \bs{\Lambda}_{p-r} \widetilde{\bs{\beta}}_2\\ 
\overset{\text{by }\eqref{eq:delta-lower-bound-3}}{\leq} &\  2 \frac{\|\widetilde{S}_2^\top\bs{\Lambda}_{p-r}\widetilde{S}_2\|_2}{\lambda_{\min}(\widetilde{S}_1\widetilde{S}_1^\top)}\|\widetilde{\bs{\beta}}_1\|_2^2 + \left( 1 + 2\frac{\lambda_{\max}(\widetilde{S}_2^\top\bs{\Lambda}_{p-r}\widetilde{S}_2)}{\lambda_{\min}(\widetilde{S}_2^\top\bs{\Lambda}_{p-r}\widetilde{S}_2)} \cdot
\frac{\lambda_{\max}(\widetilde{S}_1\widetilde{S}_1^\top)}{\lambda_{\min}(\widetilde{S}_1\widetilde{S}_1^\top)}
\right) \cdot \widetilde{\bs{\beta}}_2^\top \bs{\Lambda}_{p-r} \widetilde{\bs{\beta}}_2\\
=~  & \ 2 \frac{\|\widetilde{S}_2^\top\bs{\Lambda}_{p-r}\widetilde{S}_2\|_2}{\lambda_{\min}(\widetilde{S}_1\widetilde{S}_1^\top)}\|\widetilde{\bs{\beta}}_1\|_2^2 + \left( 1 + 2 \kappa(\widetilde{S}_2^\top\bs{\Lambda}_{p-r}\widetilde{S}_2) \kappa(\widetilde{S}_1\widetilde{S}_1^\top) \right) \cdot \widetilde{\bs{\beta}}_2^\top \bs{\Lambda}_{p-r} \widetilde{\bs{\beta}}_2.
\end{align*}
This completes the proof of Lemma~\ref{lem:delta-bound-1}.

\paragraph*{Proof of \eqref{eq:delta-lower-bound-2} and \eqref{eq:delta-lower-bound-3}.} First we show \eqref{eq:delta-lower-bound-2}. By the triangle inequality, we have 
\[
\|\widetilde{\bs{\beta}}_1 - \widetilde{S}_1 (\widetilde{S}_2^\top \bs{\Lambda}_{p-r} \widetilde{S}_2)^{-1} \widetilde{S}_2^\top \bs{\Lambda}_{p-r}\widetilde{\bs{\beta}}_2\|_2^2 \leq 2 \|\widetilde{\bs{\beta}}_1 \|_2^2 + 2\|\widetilde{S}_1 (\widetilde{S}_2^\top \bs{\Lambda}_{p-r} \widetilde{S}_2)^{-1} \widetilde{S}_2^\top \bs{\Lambda}_{p-r}\widetilde{\bs{\beta}}_2\|_2^2.
\]
Note that for $\bm{A}\in\mathbb{R}^{p\times p}$ and $\bm{x}\in\mathbb{R}^p$, the multiplicative property of the norm shows $\|\bm{Ax}\|_2 \leq \|\bm{A}\|_2\|\bm{x}\|_2$. Using this property, it can be seen that  
\[
\|\widetilde{S}_1 (\widetilde{S}_2^\top \bs{\Lambda}_{p-r} \widetilde{S}_2)^{-1} \widetilde{S}_2^\top \bs{\Lambda}_{p-r}\widetilde{\bs{\beta}}_2\|_2^2 \leq \|\widetilde{S}_1\|_2^2 \|(\widetilde{S}_2^\top \bs{\Lambda}_{p-r} \widetilde{S}_2)^{-1} \widetilde{S}_2^\top \bs{\Lambda}_{p-r}^{1/2}\|_2^2 \|\bs{\Lambda}_{p-r}^{1/2}\widetilde{\bs{\beta}}_2\|_2^2.
\]
By $\|\bm{AA}^\top\|_2 = \|\bm{A}\|_2^2$, it follows that
\[
\|(\widetilde{S}_2^\top \bs{\Lambda}_{p-r} \widetilde{S}_2)^{-1} \widetilde{S}_2^\top \bs{\Lambda}_{p-r}^{1/2}\|_2^2 = \|(\widetilde{S}_2^\top \bs{\Lambda}_{p-r} \widetilde{S}_2)^{-1}\|_2.
\]
Then we have
\begin{equation*}
\|\widetilde{\bs{\beta}}_1 - \widetilde{S}_1 (\widetilde{S}_2^\top \bs{\Lambda}_{p-r} \widetilde{S}_2)^{-1} \widetilde{S}_2^\top \bs{\Lambda}_{p-r}\widetilde{\bs{\beta}}_2\|_2^2 \leq 2 \|\widetilde{\bs{\beta}}_1 \|_2^2 + 2 \|\widetilde{S}_1\|_2^2 \|(\widetilde{S}_2^\top \bs{\Lambda}_{p-r} \widetilde{S}_2)^{-1}\|_2.
\end{equation*}
It remains to show \eqref{eq:delta-lower-bound-3}.
By definition, for a symmetric matrix $\bm{A}$, we can write $\lambda_{\min}(\bm{A}) = \displaystyle{\min_{\|\bm{x}\|_2=1}} \bm{x}^\top\bm{A}\bm{x}$. Taking $\forall\bm{x}\in\mathbb{R}^r$ with $\|\bm{x}\|_2 = 1$, we have
\begin{align*}
\bm{x}^\top \bm{U}_r^\top \bs{H} \bm{U}_r \bm{x} & = \bm{x}^\top \bm{U}_r^\top S_k (S_k^\top \cov_{p-r} S_k)^{-1}S_k^\top \bm{U}_r \bm{x}  = \bm{x}^\top \widetilde{S}_1 (\widetilde{S}_2^\top \bs{\Lambda}_{p-r} \widetilde{S}_2)^{-1}\widetilde{S}_1^\top \bm{x} \\[.5em]
& \geq \lambda_{\min}((\widetilde{S}_2^\top \bs{\Lambda}_{p-r} \widetilde{S}_2)^{-1})(\bm{x}^\top \widetilde{S}_1\widetilde{S}_1^\top \bm{x}) \\[.5em]
& \geq \lambda_{\min}((\widetilde{S}_2^\top \bs{\Lambda}_{p-r} \widetilde{S}_2)^{-1}) \lambda_{\min}(\widetilde{S}_1\widetilde{S}_1^\top) = \frac{\lambda_{\min}(\widetilde{S}_1\widetilde{S}_1^\top)}{\lambda_{\max}(\widetilde{S}_2^\top \bs{\Lambda}_{p-r} \widetilde{S}_2)}.
\end{align*}
Then we know
\begin{equation*}
\|(\bm{U}_r^\top \bs{H} \bm{U}_r)^{-1}\|_2
= \frac{1}{\lambda_{\min}(\bm{U}_r^\top \bs{H} \bm{U}_r)}
\leq \frac{\lambda_{\max}(\widetilde{S}_2^\top \bs{\Lambda}_{p-r} \widetilde{S}_2)}{\lambda_{\min}(\widetilde{S}_1\widetilde{S}_1^\top)}.
\end{equation*}

\subsection{Proof of Lemma~\ref{lem:matrix-norm}}
\label{SecPfMatrix-norm}
Let us write the singular value decomposition of $\cov$ as $\cov = \bs{U}\bs{\Lambda}\bs{U}^\top $ with $\bs{\Lambda} = \text{diag}(\bs{\lambda})$, $\bs{\lambda}\in\mathbb{R}^p$. Then we have $\text{tr}(\cov) = \|\bs{\lambda}\|_1$, $\|\cov\|_F = \|\bs{\lambda}\|_2$, $\text{tr}(\cov^2) = \|\bs{\lambda}\|_2^2$ and $\text{tr}(\cov^4) = \|\bs{\lambda}\|_4^4$. With the new notation, the claim of Lemma~\ref{lem:matrix-norm} is now equivalent to $\|\bs{\lambda}\|_1^2 \|\bs{\lambda}\|_4 \geq \|\bs{\lambda}\|_2^3$. 

We prove $\|\bs{\lambda}\|_1^2 \|\bs{\lambda}\|_4 \geq \|\bs{\lambda}\|_2^3$ using the following ingredients:
\begin{align*}
&\text{(i)}~ \|\bs{\lambda}\|_3^3\|\bs{\lambda}\|_1 \geq \|\bs{\lambda}\|_2^4; \\[.5em]
&\text{(ii)}~ \|\bs{\lambda}\|_4^4\|\bs{\lambda}\|_1 \geq \|\bs{\lambda}\|_3^3 \|\bs{\lambda}\|_2^2; \\[.5em]
&\text{(iii)}~ \|\bs{\lambda}\|_1 \geq \|\bs{\lambda}\|_2,
\end{align*}
where (i) holds directly from Cauchy-Schwarz inequality; (ii) follows from the equality 
\begin{align*}
\|\bs{\lambda}\|_4^4\|\bs{\lambda}\|_1 - \|\bs{\lambda}\|_3^3 \|\bs{\lambda}\|_2^2 = \frac{1}{2}\sum_{i\neq j} \lambda_i\lambda_j(\lambda_i+\lambda_j)(\lambda_i-\lambda_j)^2 \geq 0
\end{align*}
with $\lambda_i\geq 0$; (iii) follows from the observation that $\|\bs{\lambda}\|_1^2 - \|\bs{\lambda}\|_2^2 = \sum_{i\neq j}\lambda_i\lambda_j \geq 0$ with $\lambda_i\geq 0$.

Then we have 
\[
\|\bs{\lambda}\|_1^8 \|\bs{\lambda}\|_4^4 = \left(\frac{\|\bs{\lambda}\|_1 \|\bs{\lambda}\|_4^4}{\|\bs{\lambda}\|_3^3} \right) \cdot (\|\bs{\lambda}\|_1 \|\bs{\lambda}\|_3^3) \cdot (\|\bs{\lambda}\|_1^6) \geq 
\|\bs{\lambda}\|_2^2 \cdot \|\bs{\lambda}\|_2^4 \cdot \|\bs{\lambda}\|_2^6 = \|\bs{\lambda}\|_2^{12}.
\]
Thus we show $\|\bs{\lambda}\|_1^2 \|\bs{\lambda}\|_4 \geq \|\bs{\lambda}\|_2^3$, and Lemma~\ref{lem:matrix-norm} follows.

\subsection{Proof of Lemma~\ref{lem:minimal-eigenvalue}}
We closely follow the proof of Theorem 5.39 in \cite{vershynin2010introduction} that uses a covering argument with three steps: 1) discretization; 2) concentration; 3) union bound. In the discretization step, we discretize the problem using a net $\mathcal{N}$; in the concentration step, we bound $\|\bs{Ax}\|_2$ for each $\bs{x}\in\mathcal{N}$. Finally, we use the union bound to establish a concentration bound over $\bs{x}\in\mathcal{S}^{n-1}$. 

\paragraph{Step 1: Discretization.} First we invoke Lemma 5.36 in \cite{vershynin2010introduction}:
\begin{lemma}\label{lem:isometry}
	Consider a matrix $\bs{B}$ that satisfies
	\[
	\|\bs{B}^\top \bs{B} - \emph{\Ind} \|_2 \leq \max(\delta,\delta^2)
	\]
	for some $\delta > 0$. Then 
	\[
	1-\delta \leq s_{\min}(\bs{B}) \leq s_{\max}(\bs{B}) \leq 1 + \delta.
	\]
	Conversely, if $\bs{B}$ satisfies $1-\delta \leq s_{\min}(\bs{B}) \leq s_{\max}(\bs{B}) \leq 1 + \delta$ for some $\delta >0$, then $\|\bs{B}^\top \bs{B}-\emph{\Ind}\|_2 \leq 3\max(\delta, \delta^2)$.
\end{lemma}
Write $T = \|\bs{\Lambda}\|_2^2$ and $\bs{A}=\bs{\Lambda} \bs{S}$. Then the claim is equivalent to 
\[
\Big\|\frac{1}{T} \bs{A}^\top \bs{A}-\Ind \Big\|_2 \leq \max(t, t^2) = t.
\]
We can evaluate the operator norm on a $1/4$-net $\mathcal{N}$ of the unit sphere $\mathcal{S}^{n-1}$: with Lemma 5.4 in \cite{vershynin2010introduction}, we have
\[
\Big\|\frac{1}{T} \bs{A}^\top \bs{A}-\Ind \Big\|_2  \leq 2 \max_{x\in\mathcal{N}} \left| \frac{1}{T}\|\bs{Ax}\|_2^2-1\right|.
\]
Note that we can choose $\mathcal{N}$ such that $|\mathcal{N}|\leq 9^n$.

\paragraph{Step 2: Concentration.} Fix $\bs{x}\in\mathcal{S}^{n-1}$. 
Denote the $i$-th row of matrix $\bs{A}$ and $\bs{S}$ by $\bs{A}_i$ and $\bs{S}_i$, respectively. Then $\left\langle \bs{A}_i, \bs{x}\right\rangle /\lambda_i = \left\langle \bs{S}_i, \bs{x}\right\rangle \sim \mathcal{N}(0, 1)$ and the $\left\langle \bs{A}_i, \bs{x}\right\rangle$'s are independent to each other. We can express $\|\bs{Ax}\|_2^2$ as a sum of independent random variables
\[
\|\bs{Ax}\|_2^2 = \sum_{i=1}^N \left\langle \bs{A}_i, \bs{x}\right\rangle ^2 =:\sum_{i=1}^N \lambda_i^2 Z_i^2,
\]
where $Z_i\overset{iid}{\sim}\mathcal{N}(0, 1)$. By Lemma 1 of \cite{laurent2000adaptive}, we have
\[
P\left(\left|\frac{1}{\sum_{i=1}^N \lambda_i^2}\|\bs{Ax}\|_2^2-1\right| \geq
2 \frac{\sqrt{\sum_{i=1}^N \lambda_i^4}}{\sum_{i=1}^N \lambda_i^2} \sqrt{\delta} + 2 \frac{\max_{1\leq i\leq N} \lambda_i^2}{\sum_{i=1}^N \lambda_i^2 } \delta
\right) \leq 2e^{-\delta}. 
\]
When $\delta=\min\left\lbrace
\frac{1}{16} \frac{\|\lambda\|_2^4}{\|\lambda\|^4_4} t^2 , \frac{1}{4} \frac{\|\lambda\|_2^2}{\|\lambda\|_{\infty}^2} t \right\rbrace$, we have
$2 \frac{\sqrt{\sum_{i=1}^N \lambda_i^4}}{\sum_{i=1}^N \lambda_i^2} \sqrt{\delta}\leq \frac{1}{2}t$ and $2 \frac{\max_{1\leq i\leq N} \lambda_i^2}{\sum_{i=1}^N \lambda_i^2 } \delta \leq \frac{1}{2} t$. Then
we can rewrite the tail bound as 
\[
P\left(\left|\frac{1}{\sum_{i=1}^N \lambda_i^2}\|\bs{Ax}\|_2^2-1\right| \geq
t\right) \leq 2\exp \left(-\min\left\lbrace
\frac{1}{16} \frac{\|\bs{\lambda}\|_2^4}{\|\bs{\lambda}\|^4_4} t^2 , \frac{1}{4} \frac{\|\bs{\lambda}\|_2^2}{\|\bs{\lambda}\|_{\infty}^2} t \right\rbrace \right).
\]

\paragraph{Step 3: Union bound.} Taking the bound over all vectors in the net $\mathcal{N}$, we obtain
\[
P\left( \max_{x\in\mathcal{N}} \left| \frac{1}{T} \|\bs{Ax}\|_2^2 - 1\right| \geq t\right) \leq 9^n \cdot 2\exp \left(-\min\left\lbrace
\frac{1}{16} \frac{\|\bs{\lambda}\|_2^4}{\|\bs{\lambda}\|^4_4} t^2 , \frac{1}{4} \frac{\|\bs{\lambda}\|_2^2}{\|\bs{\lambda}\|_{\infty}^2} t \right\rbrace \right).
\]
Thus, by Lemma~\ref{lem:isometry}, we have, for $t<1$, 
\[
(1 - t)\sqrt{\sum_{i=1}^N \lambda_i^2} \leq s_{\min}(\bs{A}) \leq s_{\max}(\bs{A}) \leq (1 + t)\sqrt{\sum_{i=1}^N \lambda_i^2}
\]
with probability at least $1-9^n \cdot 2\exp \left( -\min\left\lbrace
\frac{1}{16} \frac{\|\bs{\lambda}\|_2^4}{\|\bs{\lambda}\|^4_4} t^2 , \frac{1}{4} \frac{\|\bs{\lambda}\|_2^2}{\|\bs{\lambda}\|_{\infty}^2} t \right\rbrace \right)$.

\subsection{Technical details of Theorem~\ref{thm:power-2}}
In this part we check some technical details of Theorem~\ref{thm:power-2}. Recall from the proof of Theorem~\ref{thm:power-function} that the sketched linear model is 
\[
y_i = \left\langle \widetilde{\bs{x}}_i, \bs{\beta}^S\right\rangle + z_i^S = \left\langle S_k^\top \bs{x}_i, \bs{\beta}^S\right\rangle + z_i^S,
\]
where $z_i^S = \left\langle \bs{x}_i, \bs{\beta}\right\rangle + \sigma z_i - \left\langle \widetilde{\bs{x}}_i, \bs{\beta}^S\right\rangle$ and $\boldsymbol{\beta}^S = (S_k^\top \boldsymbol{\Sigma} S_k)^{-1} S_k^\top \boldsymbol{\Sigma} \boldsymbol{\beta}$. We are essentially testing whether sketched coefficients $\boldsymbol{\beta}^S$ are zero or not as
$$
H_0^S: \boldsymbol{\beta}^S=0 \quad \text{versus} \quad H_1^S: \boldsymbol{\beta}^S \neq 0.
$$

In what follows, we verify that the technical conditions of Theorem 2.1 and Corollary 2.2 in \cite{steinberger2016relative} are satisfied under assumptions \textbf{(B1, B2)} and the sketched model $y_i = \left\langle \widetilde{\bs{x}}_i, \bs{\beta}^S\right\rangle + z_i^S$. This verification step directly leads to the desired result in Theorem~\ref{thm:power-2}. See Section 2.1 of \cite{steinberger2016relative} for the technical conditions; specifically, it suffices to verify \textbf{(A1)}\emph{(a,b,c,d)} and \textbf{(A2)} therein. We write them as \textbf{(S-A1)}\emph{(a,b,c,d)} and \textbf{(S-A2)} below.

\paragraph{Verification of (S-A1):} By our assumption \textbf{(B1)} with $\bs{\widetilde{x}}_i = S_k^\top \bs{\Gamma}\bs{u}_i$, we can directly see assumptions \textbf{(S-A1)}\emph{(a,b,c,d)} are satisfied.

\paragraph{Verification of (S-A2):} It suffices to check the following two conditions:
\begin{align}
\label{eq:non-gaussian-1}
\mathbb{E}\left[ \left( \mathbb{E}\left[ \left( z_i^S\right)^4 \vert \widetilde{\bs{x}}_i \right]\right)^{2}\right] = O(1)
\quad \text{and} \quad
\max_{i=1}^n \mathbb{E}\left[ (z_i^S)^4 \vert \widetilde{\bs{x}}_i \right] = o_{P}(\sqrt{k}).
\end{align}

\paragraph{First claim of \eqref{eq:non-gaussian-1}.}
To simplify notation, write $\bm{\delta} \defn\bm{\beta} - S_k \bm{\beta}^S$. Then we can write 
$$z_i^S = \sigma z_i + \bm{\delta}' \bm{x}_i.$$ 
We first derive the expression for $\mathbb{E}[(z_i^S)^4 \vert \widetilde{\bs{x}}_i]$. Notice that $\mathbb{E}[(z_i^S)^4 \vert \widetilde{\bs{x}}_i] = \mathbb{E}[\mathbb{E}[(z_i^S)^4 \vert \bs{x}_i]\bigm\arrowvert \widetilde{\bs{x}}_i ]$, with
\begin{equation}\label{eq:non-gaussian-3}
\mathbb{E}[\left(z_i^S\right)^4 \vert \bs{x}_i] = \mathbb{E}\left[
\left(\sigma z_i + \bm{\delta}' \bm{x}_i\right)^4 \vert \bs{x}_i  \right] \leq 8c \sigma^4 + 8 \left(\bm{\delta}' \bm{x}_i\right)^4.
\end{equation}
The above inequality follows by $(x+y)^4 \leq 8(x^4+y^4)$ as well as assumption \textbf{(B2)}. Then we further have
\begin{equation}\label{eq:non-gaussian-2}
\mathbb{E}\left[ \left( \mathbb{E}\left[ \left( z_i^S\right)^4 \vert \widetilde{\bs{x}}_i \right]\right)^{2}\right] = \mathbb{E}\left[ \left(
8c \sigma^4 + 8 \mathbb{E}\left[ \left(\bm{\delta}' \bm{x}_i\right)^4 \vert \widetilde{\bs{x}}_i \right] \right)^{2}\right] \leq 128\left(c^2 \sigma^8 + \mathbb{E}\left[ \left( \mathbb{E}\left[ \left( \bm{\delta}' \bm{x}_i \right)^4 \vert \widetilde{\bs{x}}_i \right]\right)^{2}\right] \right).
\end{equation}
To show the first claim in \eqref{eq:non-gaussian-1}, it suffices to show
$\mathbb{E}[ ( \mathbb{E}[ ( \bm{\delta}' \bm{x}_i )^4 \vert \widetilde{\bs{x}}_i ])^{2}] = O(1) $. By $\text{Var}(\mathbb{E}[Y|X]) \leq \text{Var}(Y)$, we have
\begin{align*}
\mathbb{E}\left[ \left( \mathbb{E}\left[ \left( \bm{\delta}' \bm{x}_i \right)^4 \vert \widetilde{\bs{x}}_i \right]\right)^{2}\right] & = \text{Var}\left(
\mathbb{E}\left[ \left( \bm{\delta}' \bm{x}_i\right)^4 \vert \widetilde{\bs{x}}_i \right]
\right) + \left(\mathbb{E}\left[\mathbb{E}\left[ \left( \bm{\delta}' \bm{x}_i \right)^4 \vert \widetilde{\bs{x}}_i \right]
\right]\right)^2 \\
& \leq \text{Var}\left( \left( \bm{\delta}' \bm{x}_i \right)^4\right) + 
\left(\mathbb{E}\left[ \left( \bm{\delta}' \bm{x}_i \right)^4 \right]\right)^2
= \mathbb{E}\left[\left( \bm{\delta}' \bm{x}_i \right)^8 \right] .
\end{align*}
With $\bm{\delta} =\bm{\beta} - S_k \bm{\beta}^S$, we also have
\begin{equation}\label{eq:non-gaussian-4}
\mathbb{E}\left[\left(  \bm{\delta}' \bm{x}_i \right)^8 \right] = \mathbb{E}\left[\left\langle \bm{x}_i, \bm{\beta} - S_k \bm{\beta}^S \right\rangle^8\right]
\leq \|\bm{\Gamma}^\top(\bm{\beta} - S_k \bm{\beta}^S)\|_2^8 \sup_{\|v\|_2=1}(\mathbb{E}|v'\bs{u}_i|^8).
\end{equation}
By definition of $\bm{\beta}^S$, we know $\|\bm{\Gamma}^\top(\bm{\beta} - S_k \bm{\beta}^S)\|_2^2 = \bm{\beta}^\top \cov \bm{\beta}-\Delta_k^2 \leq \bm{\beta}^\top \cov \bm{\beta} = o(1)$.
By \textbf{(B1)}(b), we further know $\sup_{\|v\|=1}(\mathbb{E}|v'\bs{u}_i|^8)=O(1)$. Thus we show $\mathbb{E}[ ( \mathbb{E}[ ( \bm{\delta}' \bm{x}_i )^4 \vert \widetilde{\bs{x}}_i ])^{2}] = O(1) $. Therefore, together with inequality~\eqref{eq:non-gaussian-2}, the first claim in \eqref{eq:non-gaussian-1} follows.

\paragraph{Second claim of \eqref{eq:non-gaussian-1}.}
Next we show the second claim in \eqref{eq:non-gaussian-1}. By inequality \eqref{eq:non-gaussian-3}, we have
\[
\max_{i=1}^n \mathbb{E}\left[ (z_i^S)^4 \vert \widetilde{\bs{x}}_i \right] \leq 8c\sigma^4 + 8 \max_{i=1}^n \mathbb{E}\left[ (\bm{\delta}' \bm{x}_i)^4 \vert \widetilde{\bs{x}}_i \right],
\]
and it suffices to show that $\max_{i=1}^n \mathbb{E}\left[ (\bm{\delta}' \bm{x}_i)^4 \vert \widetilde{\bs{x}}_i \right]= o_{P}(\sqrt{k})$. Observe that
\begin{align*}
\mathbb{P}\left(\max_{i=1}^n \mathbb{E}\left[ (\bm{\delta}' \bm{x}_i)^4 \vert \widetilde{\bs{x}}_i \right] \geq \epsilon \right) \overset{\text{(i)}}{\leq} \frac{n}{\epsilon^2} \text{Var}\left(
(\bm{\delta}' \bm{x}_i)^4
\right) \overset{\text{(ii)}}{\leq} n \bm{\beta}^\top \cov \bm{\beta} \  \frac{\sup_{\|v\|=1}(\mathbb{E}|v'\bs{u}_i|^8)}{\epsilon^2} \overset{\text{(iii)}}{=} o\left(\frac{k}{\epsilon^2}\right).
\end{align*}
In the above argument, step~(i) follows from the union bound and Chebyshev's inequality; step~(ii) is from \eqref{eq:non-gaussian-4} and $\text{Var}\left((\bm{\delta}' \bm{x}_i)^4 \right) \leq \mathbb{E}[\left(  \bm{\delta}' \bm{x}_i \right)^8 ]$; step (iii) uses the local alternative $\bm{\beta}^\top \cov \bm{\beta} = o(k/n)$ and assumption \textbf{(B1)}(b). Therefore we can conclude that $\max_{i=1}^n \mathbb{E}[ (\bm{\delta}' \bm{x}_i)^4 \vert \widetilde{\bs{x}}_i ] = o_{P}(\sqrt{k})$, which completes the proof.

\end{document}